%% file: main-lncs.tex
\documentclass{llncs}

\usepackage{amssymb,amsmath,graphicx}
\usepackage{color,enumerate,comment,boxedminipage}
\usepackage{pgf,float,hyperref}
\usepackage{xcolor}
\definecolor{darkblue}{rgb}{0,0,0.45}
\definecolor{darkred}{rgb}{0.6,0,0}
\definecolor{darkgreen}{rgb}{0.13,0.5,0}
\hypersetup{colorlinks, linkcolor=darkblue, citecolor=darkgreen,
urlcolor=darkblue}
\usepackage[T1]{fontenc}
\usepackage{xspace}
\usepackage{paralist}
\usepackage{thmtools}
\usepackage{thm-restate}
\usepackage{cleveref}
\usepackage{enumitem}

\usepackage{amsthm}

\newtheorem{claimm}{Claim}[section]
\numberwithin{claimm}{section}

\crefname{claimm}{Claim}{Claims}

\input{macros/macros}

\input{macros/todos}

\usepackage{fullpage}

\title{Parameterized Inapproximability \newline of Independent Set in $H$-Free 
Graphs%
\thanks{An extended abstract of this paper was presented at WG 
2020~\cite{DBLP:conf/wg/DvorakFRR20}.}
}

\author{Pavel Dvo\v{r}\'{a}k\inst{1}\thanks{Supported by Czech Science Foundation 
GA{\v C}R (grant \#19-27871X).}
\and
Andreas Emil Feldmann\inst{1}$^{\star\star}$
\and
Ashutosh Rai\inst{1}\thanks{Supported by Center for Foundations of Modern Computer 
Science (Charles Univ.\ project UNCE/SCI/004).}
\and \\
Pawe{\l} Rz{\k{a}}\.{z}ewski\inst{2,3}\thanks{Supported by Polish National Science Centre grant no. 2018/31/D/ST6/00062.}
}

\institute{
Faculty of Mathematics and Physics, Charles University, Prague, Czechia\newline
\email{koblich@iuuk.mff.cuni.cz, feldmann.a.e@gmail.com, ashu.rai87@gmail.com}
\and
Faculty of Mathematics and Information Science,\\Warsaw University of Technology, 
Warsaw, Poland\\
\email{pawel.rzazewski@pw.edu.pl}
\and
University of Warsaw, Institute of Informatics, Warsaw, Poland
}

\begin{document}

\maketitle

\begin{abstract}
We study the \IS problem in $H$-free graphs, i.e., graphs excluding some fixed 
graph $H$ as an induced subgraph. We prove several inapproximability results 
both for polynomial-time and parameterized algorithms. 

Halld\'{o}rsson~[SODA~1995] showed that for every $\delta>0$ the \IS problem has a polynomial-time $(\frac{d-1}{2}+\delta)$-approximation 
algorithm in $K_{1,d}$-free graphs. We extend this result by showing that 
$K_{a,b}$-free graphs admit a polynomial-time 
$\Oh(\alpha(G)^{1-1/a})$-approximation, where $\alpha(G)$ is the size of a 
maximum independent set in $G$.
Furthermore, we complement the result of Halld\'{o}rsson by showing that 
for some $\gamma=\Theta(d/\log d),$ there is no polynomial-time $\gamma$-approximation 
algorithm for these graphs, unless \NP~=~{\sf ZPP}.

Bonnet {\em et al.}~[Algorithmica~2020] showed that \IS parameterized by the size $k$ of the independent set is \Wone-hard on graphs which do not contain 
(1)~a cycle of constant length at least~$4$, (2)~the 
star $K_{1,4}$, and (3)~any tree with two vertices of degree at least $3$ at constant distance.
We strengthen this result by proving three inapproximability results under 
different complexity assumptions for almost the same class of graphs (we weaken 
conditions (1) and (2) that $G$ does not contain a cycle of constant length at 
least 5 or $K_{1,5}$). 
First, under the ETH, there is no $f(k) \cdot n^{o(k/\log 
k)}$ algorithm for any computable function~$f$. Then, under the deterministic 
Gap-ETH, there is a constant $\delta>0$ such that no $\delta$-approximation can 
be computed in $f(k) \cdot n^{O(1)}$ time. Also, under the stronger randomized Gap-ETH 
there is no such approximation algorithm with runtime $f(k) \cdot 
n^{o(\sqrt{k})}$.

Finally, we consider the parameterization by the excluded graph $H$, and show 
that under the ETH, \IS has no $n^{o(\alpha(H))}$ algorithm in $H$-free 
graphs. Also, we prove that there is no \mbox{$d/k^{o(1)}$-approximation}  
algorithm for $K_{1,d}$-free graphs
with runtime $f(d,k) \cdot n^{\Oh(1)}$, under the deterministic \mbox{Gap-ETH.}
\end{abstract}

\section{Introduction}\label{sec:Intro}
\input{src/intro}

\section{Preliminaries}
\input{src/prelim}

\section{Approximation for $K_{a,b}$-free Graphs}\label{sec:Kab}
\input{src/Kab}

\section{Polynomial Time Inapproximability in $K_{1,d}$-free Graphs} \label{sec:k1d}
\input{src/k_1_d_free_polynomial_inapprox}

\section{Parameterized Approximation for Fixed $H$} \label{sec:nofpas}

\input{src/fpt_apx_hardness2}

\section{Parameterized Approximation with $H$ as a Parameter} \label{sec:parameterH}
\input{src/parameter-H}

\input{src/sparsification}

\section{Conclusion and Open Problems}
\input{src/conclusions}

\section{Acknowledgement}
We would like to thank to the anonymous reviewer, who suggested using gap amplification to obtain \cref{thm:NoConstantFptApprox}.
We are also grateful to the other reviewer for pointing out the mistake in \cref{thm:noFPAS} in the conference version of our paper~\cite{DBLP:conf/wg/DvorakFRR20}.

\bibliographystyle{abbrv}
\bibliography{main}

\end{document}

%% file: macros/macros.tex
\usepackage{tabularx}

\newcommand{\Oh}{\mathcal{O}}

\newcommand{\APX}{{\sf APX}\xspace}
\newcommand{\ZPP}{{\sf ZPP}\xspace}
\newcommand{\NP}{{\sf NP}\xspace}
\newcommand{\PP}{{\sf P}\xspace}
\newcommand{\Wone}{{\sf W[1]}\xspace}

\newcommand{\Trule}{\rule{0pt}{3ex}}
\newcommand{\Brule}{\rule[-1.5ex]{0pt}{0pt}}

\usepackage{etoolbox}
\newcommand{\MCSI}[1]{%
  \ifstrempty{#1}{%
    \textsc{MCSI}\xspace
  }{%
    \textsc{MCSI(#1)}\xspace
  }%
}
\newcommand{\IS}{\textsc{Independent Set}\xspace}

\newcommand{\MIS}{\textsc{Multicolored Independent Set}\xspace}
\newcommand{\MC}{\textsc{Multicolored Clique}\xspace}

\DeclareMathOperator{\val}{val}

\newcommand{\eps}{\varepsilon}
\renewcommand{\epsilon}{\varepsilon}

\newcommand{\OptProb}[3]{
\begin{center}
\begin{tabularx}{\textwidth}{|l X|}
	\hline
	\multicolumn{2}{|c|}{{\sc #1}\Trule} \\
	\textbf{Input:\ }&{#2}\\
	\textbf{Goal:\ }&{#3\Brule}\\
	\hline
\end{tabularx}
\end{center}
}

\newcommand{\cC}{\mathcal{C}}
\renewcommand{\cref}{\Cref}

\newcommand{\repedge}{r'}
\newcommand{\repvert}[1]{r_{#1}}
\newcommand{\cH}{\mathcal{H}}

%% file: macros/todos.tex
\usepackage{todonotes}

%% file: src/intro.tex
The \IS problem, which asks for a maximum sized set of pairwise non-adjacent 
vertices in a graph, is one of the most well-studied problems in algorithmic 
graph theory. It was among the first 21 problems that were proven to be \NP-hard 
by Karp~\cite{Karp1972}, and is also known to be hopelessly difficult to 
approximate in polynomial time: H{\aa}stad~\cite{Hastad96cliqueis} proved that 
under standard  assumptions from classical complexity theory the problem admits 
no $(n^{1-\epsilon})$-approximation, for any~$\epsilon >0$ 
(by $n$ we always denote the number of vertices in the input graph). This was 
later strengthened by Khot and Ponnuswami~\cite{Khot2006}, who were able to 
exclude any algorithm with approximation ratio $n/(\log n)^{3/4 + \epsilon}$, 
for any $\epsilon > 0$. Let us point out that the currently best 
polynomial-time approximation algorithm for \IS achieves the approximation ratio 
$\Oh(n \frac{(\log \log n)^2} {(\log 
n)^3})$~\cite{DBLP:journals/siamdm/Feige04}.

There are many possible ways of approaching such a difficult problem, in order to obtain some positive results.
One could give up on generality, and ask for the complexity of the problem on 
restricted instances. For example, while the \IS problem remains \NP-hard in 
subcubic graphs~\cite{GAREY1976237}, a straightforward greedy algorithm gives a 
3-approximation.

\paragraph{$H$-free graphs.} A large family of restricted instances, for which 
the \IS problem has been well-studied, comes from forbidding certain induced 
subgraphs. For a (possibly infinite) family $\cH$ of graphs, a graph $G$ is 
\emph{$\cH$-free} if it does not contain any graph of $\cH$ as an induced 
subgraph. If $\cH$ consists of just one graph, say $\cH = \{H\}$, then we say 
that $G$ is $H$-free. 
The investigation of the complexity of \IS in $\cH$-free graphs dates back to 
Alekseev, who observed that the so-called ``Poljak construction''~\cite{Po74} yields the following.

\begin{theorem}[Alekseev~\cite{alekseev1982effect}, Poljak~\cite{Po74}]\label{thm:alekseev}
Let $s \geq 3$ be a constant. The \IS problem is \NP-hard in graphs that do
not contain any of the following induced subgraphs:
\begin{compactenum}
 \item a cycle on at most $s$ vertices,
 \item the star $K_{1,4}$, and
 \item any tree with two vertices of degree at least 3 at distance at most $s$.
\end{compactenum}
\end{theorem}

We can restate \cref{thm:alekseev} as follows: the \IS problem is 
\mbox{\NP-hard} in $H$-free graphs, unless $H$ is a subgraph of a subdivided 
claw (i.e.,~three paths which meet at one of their endpoints). The reduction 
also implies that for each such $H$ the problem is \APX-hard and cannot be 
solved in subexponential time, unless the Exponential Time Hypothesis (ETH) 
fails. On the other hand, polynomial-time algorithms are known only for very few 
cases. First let us consider the case when $H=P_t$, i.e., we forbid a path on 
$t$ vertices. Note that the case of $t=3$ is trivial, as every $P_3$-free graph 
is a disjoint union of cliques. Already in 1981 Corneil, Lerchs, and 
Burlingham~\cite{CORNEIL1981163} showed that \IS is tractable for $P_4$-free 
graphs. For many years there was no improvement, until the breakthrough 
algorithm of Lokshtanov, Vatshelle, and 
Villanger~\cite{DBLP:conf/soda/LokshantovVV14} for $P_5$-free graphs. Their 
approach later recently extended to $P_6$-free graphs by Grzesik, Klimo\v{s}ova, 
Pilipczuk, and Pilipczuk~\cite{DBLP:conf/soda/GrzesikKPP19}. 
The general belief that the problem should be polynomial-time solvable for $P_t$-free graphs,
for any fixed $t$, is suppotred by recent quasipolynomial-time algorithm by Gartland and Lokshtanov~\cite{gartlandpK}; see also a simplified version of Pilipczuk, Pilipczuk, and Rz\k{a}\.zewski~\cite{DBLP:conf/sosa/PilipczukPR21}.

Even less is known for the case if $H$ is a subdivided claw. The problem can be 
solved in polynomial time in claw-free (i.e., $K_{1,3}$-free) graphs, see 
Sbihi~\cite{SBIHI198053} and Minty~\cite{MINTY1980284}. This was later extended to $H$-free 
graphs, where $H$ is a claw with one edge once subdivided (see 
Alekseev~\cite{ALEKSEEV20043} for the unweighted version and Lozin, 
Milani\v{c}~\cite{DBLP:journals/jda/LozinM08} for the weighted one). 
We also know that for any subdivided claw~$H$, the problem can be solved in 
subexponential time in $H$-free 
graphs~\cite{DBLP:conf/soda/ChudnovskyPPT20,DBLP:conf/icalp/MajewskiM0OPRS22}.

When it comes to approximations, 
Halld{\'{o}}rsson~\cite{DBLP:conf/soda/Halldorsson95} gave an elegant local 
search algorithm that finds a $(\frac{d-1}{2}+\delta)$-approximation of a 
maximum independent set in $K_{1,d}$-free graphs for any constant $\delta>0$ 
in polynomial time. Chudnovsky, Thomass{\'e}, Pilipczuk, and 
Pilipczuk~\cite{DBLP:conf/soda/ChudnovskyPPT20} designed a QPTAS 
(quasi-polynomial-time approximation scheme) that works for subdivided claw $H$;
see also the improved version of Majewski et al.~\cite{DBLP:conf/icalp/MajewskiM0OPRS22}.
Recall that if $H$ is not (a subgraph of) a subdivided claw, then the problem is \APX-hard.
The existence of algorithms for MIS in $H$-free graph with approximation guararantee $n^{1-\delta}$ for constant $\delta$ was 
studied recently by Bonnet et al.~\cite{DBLP:conf/esa/BonnetTTW20} in the connection to the
famous Erd\H{o}s-Hajnal conjecture.

\paragraph{Parameterized complexity.} Another approach that one could take is 
to look at the problem from the parameterized perspective: we no longer insist 
on finding a maximum independent set, but want to verify whether some 
independent set of size at least $k$ exists. To be more precise, we are 
interested in knowing how the complexity of the problem depends on $k$. The best type 
of behavior we are hoping for is \emph{fixed-parameter tractability} (FPT), 
i.e., the existence of an algorithm with running time $f(k) \cdot n^{\Oh(1)}$, 
for some function $f$ (note that since the problem is \NP-hard, we expect $f$ to 
be super-polynomial).

It is known~\cite{DBLP:books/sp/CyganFKLMPPS15} that on general graphs the \IS 
problem is \Wone-hard parameterized by $k$, which is a strong indication that 
it does not admit an FPT algorithm. Furthermore, it is even unlikely to admit 
any non-trivial \emph{fixed-parameter approximation (FPA)}: a $\beta$-FPA 
algorithm (for $\beta > 1$) for the \IS problem is an algorithm that takes as input a graph $G$ 
and an integer~$k$, and in time $f(k) \cdot n^{\Oh(1)}$ either correctly 
concludes that $G$ has no independent set of size at least~$k$, or outputs an 
independent set of size at least $k/\beta$ (note that~$\beta$ does not have 
to be a constant). It was shown in~\cite{DBLP:conf/focs/ChalermsookCKLM17} that 
on general graphs no $o(k)$-FPA exists for \IS, unless the 
deterministic\footnote{While this is stated under the randomized Gap-ETH 
in~\cite{DBLP:conf/focs/ChalermsookCKLM17}, a derandomization exists; 
see~\cite[Section~4.2.1]{DBLP:conf/focs/ChalermsookCKLM17}.} Gap-ETH fails.

\paragraph{Parameterized complexity in $H$-free graphs.}
As we pointed out, none of the discussed approaches, i.e., considering $H$-free 
graphs or considering parameterized algorithms, seems to make the \IS problem 
more tractable. However, some positive results can be obtained 
by combining these two settings, i.e., considering the parameterized complexity 
of \IS in $H$-free graphs.
For example, the Ramsey theorem implies that any graph with $\Omega(4^{p})$ 
vertices contains a clique or an independent set of size $\Omega(p)$.
Since the proof actually tells us how to construct a clique or an independent set in polynomial time~\cite{Erdos1987}, we immediately obtain a very simple FPT algorithm for $K_p$-free graphs.
Dabrowski~\cite{DBLP:conf/iwoca/DabrowskiLMR10} provided some positive and negative results for the complexity of the \IS problem in $H$-free graphs, for various $H$.
The systematic study of the problem was initiated by Bonnet, Bousquet, Charbit, Thomass\'{e}, and Watrigant~\cite{DBLP:conf/iwpec/BonnetBCTW18} and continued by Bonnet, Bousquet, Thomass\'{e}, and Watrigant~\cite{DBLP:conf/isaac/BonnetBTW19}.
Among other results,  Bonnet {\em et al.}~\cite{DBLP:conf/iwpec/BonnetBCTW18} obtained the following 
analog of \cref{thm:alekseev}.

\begin{theorem}[Bonnet {\em et al.}~\cite{DBLP:conf/iwpec/BonnetBCTW18}]\label{thm:bonnet-w1hard}
Let $s \geq 4$ be a constant. The \IS problem is  \Wone-hard in graphs that do
not contain any of the following induced subgraphs:
\begin{compactenum}
 \item a cycle on at least~4 and at most $s$ vertices,
 \item the star~$K_{1,4}$, and
 \item any tree with two vertices of degree at least 3 at distance at most~$s$.
\end{compactenum}
\end{theorem}

Note that, unlike in \cref{thm:alekseev}, we are not able to show hardness for 
$C_3$-free graphs: as already mentioned, the Ramsey theorem implies that 
\IS is FPT in $C_3$-free graphs. Thus, graphs $H$ for which there is hope for 
FPT algorithms in $H$-free graphs are essentially obtained from paths and 
subdivided claws (or their subgraphs) by replacing each vertex with a clique.

Let us point out that, even though it is not stated there explicitly, the reduction of 
Bonnet {\em et al.}~\cite{DBLP:conf/iwpec/BonnetBCTW18} also excludes any algorithm solving the problem in time $f(k) \cdot n^{o(\sqrt{k})}$, unless the ETH fails.

\paragraph{Our results.}
We study the approximation of the \IS problem in 
$H$-free graphs, mostly focusing on approximation hardness.
Our first two results are related to 
Halld\'orsson's~\cite{DBLP:conf/soda/Halldorsson95} polynomial-time 
$(\frac{d-1}{2}+\delta)$-approximation algorithm for \mbox{$K_{1,d}$-free} 
graphs.
First, in \cref{sec:Kab} we extend this result to $K_{a,b}$-free graphs, for any constants $a,b$, showing the following theorem.

\begin{restatable}{theorem}{kabtheorem}
\label{thm:Kab}
Given a $K_{a,b}$-free graph $G$, an 
$\Oh\bigl((a+b)^{1/a} \cdot \alpha(G)^{1-1/a}\bigr)$-approx\-imation can be computed in 
$n^{\Oh(a)}$ time.
\end{restatable}

Then, in~\cref{sec:k1d} we show that the approximation ratio of the algorithm  of Halld{\'o}rsson~\cite{DBLP:conf/soda/Halldorsson95} is optimal, up to logarithmic factors.

\begin{restatable}{theorem}{starfreepoly}\label{thm:k1dfree}
There is a constant $d^\star$ and a function $\beta=\Theta(d/\log d)$ such 
that for any $d\geq d^\star$ the \IS problem does not admit a polynomial time 
$\beta$-approximation algorithm in $K_{1,d}$-free graphs, unless \ZPP~=~\NP.
\end{restatable}

We remark that for graphs of maximum degree~$d$, which are also $K_{1,d}$-free, 
\IS admits a polynomial time $\tilde{O}(d/\log^2 
d)$-approximation~\cite{bansal2018lovasz} 
(where the $\tilde{O}$-notation hides poly($\log\log d$) factors) and this is 
tight~\cite{austrin2011inapproximability} (up to poly($\log\log d$) factors) 
under the Unique Games Conjecture. Also, assuming \PP~$\neq$~\NP no $\Oh(d/\log^4 
d)$-approximation exists~\cite{chan2016approximation}. This means that the 
hardness of \cref{thm:k1dfree} together with the algorithm 
in~\cite{bansal2018lovasz} give a separation between graphs of maximum degree 
$d$ and $K_{1,d}$-free graphs in terms of approximation.

Then in \cref{sec:nofpas} we study the existence of fixed-parameter 
approximation algorithms (cf.~\cite{DBLP:journals/algorithms/FeldmannSLM20}) for 
the \IS problem in $H$-free graphs. We show the following strengthening of 
\cref{thm:bonnet-w1hard}, which also gives (almost) tight runtime lower bounds 
assuming the ETH or the randomized Gap-ETH (for more information about 
complexity assumptions used in the following theorems see \cref{sec:prelims}).

\begin{restatable}{theorem}{noFPAS}\label{thm:noFPAS}
Let $s \geq 5$ be a constant, and let $\mathcal{G}$ be the class of graphs that
do not contain any of the following induced subgraphs:
\begin{compactenum}
 \item a cycle on at least~5 and at most $s$ vertices, 
 \item the star~$K_{1,5}$, and
 \item 
 \begin{compactenum}[(i)]
 \item the star $K_{1,4}$, or
 \item a cycle on 4 vertices and any tree with two 
vertices of degree at least 3 at distance at most~$s$.
 \end{compactenum}
\end{compactenum}
The \IS problem on $\mathcal{G}$ does not admit the following:
\begin{compactenum}[(a)]
\item an exact algorithm with runtime $f(k) \cdot n^{o(k/\log k)}$, for any 
computable function $f$, under the ETH,
\item a $\beta$-approximation algorithm  with runtime $f(k) \cdot n^{\Oh(1)}$ 
for some constant $\beta>1$ and any computable function $f$, under the 
deterministic Gap-ETH, %
\label{it:NoFptApprox}
\item a $\beta$-approximation algorithm with runtime $f(k) \cdot n^{o(\sqrt{k})}$ for 
some constant $\beta>1$ and any computable function $f$, under the randomized 
Gap-ETH.\footnote{In the conference version of this 
paper~\cite{DBLP:conf/wg/DvorakFRR20} we mistakenly claimed our reduction 
excludes an algorithm with running time $f(k) 
\cdot n^{o(k)}$.} %
\end{compactenum}
\end{restatable}

By gap amplification using the lexicographical graph product, we are able to 
strengthen statement \eqref{it:NoFptApprox} of \cref{thm:noFPAS}, but we need 
to consider a larger class of graphs to obtain the lower bound.
We say two vertices $u, v$ are \emph{twins} if, apart from the adjacency between them, their neighborhoods are the 
same, i.e., $N(u) \setminus \{v\} = N(v) \setminus \{u\}$.
\begin{restatable}{theorem}{NoConstantFptApprox}\label{thm:NoConstantFptApprox}
 Let $s \geq 5$ be a constant, and let $\mathcal{G}'$ be the class of graphs that
do not contain any of the following induced subgraphs:
\begin{compactenum}
 \item a cycle on at least 5 and at most $s$ vertices, and
 \item any tree without twins and with two vertices of degree at least 3 at distance at most~$s$.
\end{compactenum}
 Then for any constant $\beta > 1$, the \IS problem on $\mathcal{G}'$ does not admit a $\beta$-approximation algorithm  with runtime $f(k) \cdot n^{\Oh(1)}$ for any computable function $f$, under the deterministic Gap-ETH.
\end{restatable}

In contrast with statement \eqref{it:NoFptApprox} of \cref{thm:noFPAS}, \cref{thm:NoConstantFptApprox} refuses any constant-factor FPT approximation for the \IS problem on the class $\mathcal{G}'$.
However, the gap amplification works only for forbidden graphs without twins.
Thus, the class $\mathcal{G}'$ is larger than the class $\mathcal{G}$ defined in \cref{thm:noFPAS}, and the family of forbidden subgraphs of $\mathcal{G}'$ is exactly the family of forbidden subgraphs of $\mathcal{G}$ restricted to graphs without twins.
We prove \cref{thm:NoConstantFptApprox}  in \cref{sec:noconstant}.

Finally, in \cref{sec:parameterH} we study a slightly different setting, where 
the graph~$H$ is not considered to be fixed.
As mentioned before, \IS is known to be polynomial-time solvable in $P_t$-free 
graphs for $t \leq 6$. The algorithms for increasing values of $t$ get 
significantly more complicated and their complexity increases. Thus it is 
natural to ask whether this is an inherent property of the problem and can be 
formalized by a runtime lower bound when parameterized by $t$.

We give an affirmative answer to this question, even if the forbidden family is
not a family of paths: note that the independent set number $\alpha(P_t)$ of a path on $t$ 
vertices is $\lceil t/2\rceil$. 

\begin{restatable}{proposition}{parameterH}\label{thm:parameterH}
For any integer $d$, let $\mathcal{H}_d$ be a class of graphs so that 
$\alpha(H) > d$ for every $H \in \mathcal{H}_d$, and let $\zeta$ be 
any function in $\omega(1)$.
Consider an instance~$(G,k)$ of \IS and let~$d$ be the minimum value for which 
$G$ is $\mathcal{H}_d$-free.
The \IS problem is \Wone-hard parameterized by $d$ and cannot be solved in 
$n^{o(d)}$ time, unless the ETH fails. Furthermore, no $d^{o(1)}$-approximation 
can be computed in $f(d)n^{\Oh(1)}$ time under ETH, and no independent set of 
size~$\zeta(d)$ can be computed in $f(d)n^{\zeta(d)}$ time under the 
deterministic Gap-ETH.
\end{restatable}

We also study the special case when $H=K_{1,d}$ and consider the 
inapproximability of the problem parameterized by both $\alpha(K_{1,d})=d$ and 
$k$. Unfortunately, for the parameterized version we do not obtain a clear-cut 
statement as in \cref{thm:k1dfree}, since in the following theorem $d$ cannot 
be chosen independently of $k$ in order to obtain an inapproximability gap.

\begin{restatable}{proposition}{starsparsified}\label{thm:starsparsified}
Let $\eps>0$ be any constant, $\xi(k)=2^{(\log k)^{1/2+\eps}}$, and $\zeta$ be 
any function in $\omega(1)$.
The \IS problem in $K_{1,d}$-free graphs has no $d/\xi(k)$- and 
no $d/\zeta(k)$-approx\-imation algorithm with runtime $f(d,k) \cdot n^{\Oh(1)}$ 
for any computable function~$f$, unless the deterministic Gap-ETH or the 
Strongish Planted Clique Hypothesis fails, respectively.
\end{restatable}

Note that this in particular shows that if we allow $d$ to grow as a 
polynomial~$k^\eps$ for any constant~$0< \eps<1/2$, then no 
$k^{\delta}$-approximation is possible for any~$\delta<\eps$ (since 
$\xi(k)=k^{o(1)}$), under the deterministic Gap-ETH. Under the Strongish 
Planted Clique Hypothesis, we can even allow $d$ to grow arbitrarily slowly in 
$k$ and still get an approximation lower bound. This indicates that the 
$(\frac{d-1}{2}+\delta)$-approximation for \mbox{$K_{1,d}$-free} 
graphs~\cite{DBLP:conf/soda/Halldorsson95} is likely to be best possible (up to 
sub-polynomial factors), even when parameterizing by $k$ and~$d$.
The proofs of \cref{thm:parameterH} and \cref{thm:starsparsified} can be found in~\cref{sec:parameterH}.

%% file: src/prelim.tex
\label{sec:prelims}

All our hardness results for \IS are obtained by reductions from some variant 
of the \textsc{Maximum Colored Subgraph Isomorphism} (\MCSI{}) problem. This 
optimization problem has been widely studied in the literature, both to obtain 
polynomial-time and parameterized inapproximability results, but also in its 
decision version to obtain parameterized runtime lower bounds.
We note that by applying standard transformations, \MCSI{} contains 
the well-known problems \textsc{Label Cover}~\cite{laekhanukit2014parameters} 
and \textsc{Binary CSP}~\cite{Lokshtanov17}: for \textsc{Binary CSP} the graph 
$J$ is a complete graph, while for \textsc{Label Cover} $J$ is usually 
bipartite.

\OptProb{Maximum Colored Subgraph Isomorphism (\MCSI{})}
{A graph $G$, whose vertex set is partitioned into subsets $V_1,\dots,V_\ell$, 
and a graph $J$ on vertex set $\{1,\dots,\ell\}$.}
{Find an assignment $\phi: V(J) \to V(G)$, where $\phi(i) \in V_i$ for every $i 
\in [\ell]$, that maximizes the number $S(\phi)$ of satisfied edges, i.e.,
\newline $S(\phi) := \bigl|\bigl\{ ij \in E(J) \;|\; \phi(i)\phi(j) \in 
E(G)\bigr\}\bigr|.$
}

Given an instance $\Gamma=(G,V_1,\ldots,V_\ell,J)$ of \MCSI{}, we refer to the 
number of vertices of $G$ as the \emph{size} of~$\Gamma$.  Any assignment 
$\phi: 
V(J) \to V(G)$, such that for every $i$ it holds that $\phi(i) \in V_i$, is 
called a \emph{solution of~$\Gamma$}. The \emph{value} of a solution~$\phi$ is 
$\val(\phi) := S(\phi)/|E(J)|$, i.e., the fraction of satisfied edges. The value 
of the instance $\Gamma$, denoted by $\val(\Gamma)$, is the maximum value of any 
solution of $\Gamma$.

When considering the decision version of \MCSI{}, i.e., determining whether 
$\val(\Gamma)=1$ or $\val(\Gamma)<1$, a classic hardness result for \MC
(i.e., if $J$ is a complete graph)
implies that that under the \emph{Exponential Time Hypothesis (ETH)} the problem cannot be 
solved in \mbox{$f(\ell) \cdot n^{o(\ell)}$} time for any computable function $f$.
For the optimization version of \MCSI{}, an $\alpha$-approximation is a solution 
$\phi$ with $\val(\phi)\geq 1/\alpha$. When $J$ is a complete graph, a result by 
Dinur and Manurangsi
\cite{DBLP:conf/innovations/DinurM18,DBLP:journals/corr/abs-1805-03867}
states that there is no $\ell/\xi(\ell)$-approximation algorithm, where 
$\xi(\ell)=2^{(\log\ell)^{1/2+\eps}}=\ell^{o(1)}$ for any constant $\eps>0$, 
with runtime $f(\ell) \cdot n^{O(1)}$ for any computable function~$f$, unless 
the \emph{deterministic Gap-ETH} fails (see~\cref{thm:hard-detgapETH}). This 
hypothesis assumes that there exists some constant $\delta>0$ such that no 
deterministic $2^{o(n)}$ time algorithm for \textsc{3-SAT} can decide whether 
all or at most a $(1-\delta)$-fraction of the clauses can be satisfied. A recent 
result by Manurangsi~\cite{DBLP:conf/soda/Manurangsi20} uses an even stronger 
assumption, which also rules out randomized algorithms, and in turn obtains a 
better runtime lower bound at the expense of a worse approximation lower 
bound:\footnote{The result is implicit from 
\cite[Theorem~2.1]{DBLP:conf/soda/Manurangsi20} by setting $t=2$ and using a 
straight-forward reduction from \textsc{Label Cover} to \MCSI{}, where each of 
the $\ell$ vertices of $U$ is expanded into a colour class and an edge exists if 
the respective projected labels are the same for the unique (as $t=2$) shared 
neighbor in $V$.}
when $J$ is a complete graph, there is no $\beta$-approximation algorithm for 
\MCSI{} with runtime $f(\ell) \cdot n^{o(\ell)}$ for any computable function $f$ 
and any constant $\beta$, under the \emph{randomized Gap-ETH}. This assumes 
that there exists some constant $\delta>0$ such that no randomized $2^{o(n)}$ 
time algorithm for \textsc{3-SAT} can decide whether all or at most a 
$(1-\delta)$-fraction of the clauses can be satisfied. 
Another related conjecture that was recently used to obtain lower bounds for 
\MCSI{} where $J$ is a clique, is the \emph{Strongish Planted Clique 
Hypothesis}. It states that no randomized algorithm with runtime $n^{o(\log n)}$ 
can find a planted clique of size $n^\delta$ for some $0<\delta<1/2$ in a random 
graph on $n$ vertices. Manurangsi et 
al.~\cite{manurangsi_et_al:LIPIcs.ITCS.2021.10} prove that under this 
conjecture, no $f(\ell) \cdot n^{O(1)}$ time algorithm can compute a 
$o(\ell)$-approximation to \MCSI{} (see~\cref{thm:hard-detgapETH}).

For our results we will often need the special case of \MCSI{} when the 
graph~$J$ has bounded degree. We define this problem in the following.

\OptProb{Degree-$t$ Maximum Colored Subgraph Isomorphism (\MCSI{$t$})}
{A graph $G$, whose vertex set is partitioned into subsets $V_1,\dots,V_\ell$, 
and a graph $J$ on vertex set $\{1,\dots,\ell\}$ and maximum degree $t$.}
{Find an assignment $\phi: V(J) \to V(G)$, where $\phi(i) \in V_i$ for every $i 
\in [\ell]$, that maximizes the number $S(\phi)$ of satisfied edges, i.e.,
\newline $S(\phi) := \bigl|\bigl\{ ij \in E(J) \;|\; \phi(i)\phi(j) \in 
E(G)\bigr\}\bigr|.$
}

The bounded degree case has been considered before, and we harness some of the 
known hardness results for \MCSI{$t$} in our proofs.
First, a reduction of Marx~\cite{Marx10} implies that assuming the ETH,
\MCSI{$3$} cannot be solved in time $f(\ell) \cdot n^{o(\ell/\log \ell)}$,
for any computable function $f$ (see also Marx and 
Pilipczuk~\cite[Theorem~5.5]{MarxP15}).
We also use a  polynomial-time approximation lower bound
given by Laekhanukit~\cite{laekhanukit2014parameters}, where~$t$ can be set to any 
constant and the approximation gap depends on~$t$ (see~\cref{thm:Bundit_polytime}). The complexity assumption of 
this reduction is that \NP-hard problems do not have polynomial time Las 
Vegas algorithms, i.e., \NP~$\neq$~\ZPP.
For parameterized approximations, we use a result by Lokshtanov et al.~\cite{Lokshtanov17}, who obtain a 
constant approximation gap for the case when $t=3$ (see~\cref{thm:MCSIHardness}).
It seems that this result 
for parameterized algorithms is not easily generalizable to arbitrary 
constants~$t$ so that the approximation gap would depend only on~$t$, as in the 
result for polynomial-time algorithms provided by 
Laekhanukit~\cite{laekhanukit2014parameters}: neither the techniques found 
in~\cite{laekhanukit2014parameters} nor those of~\cite{Lokshtanov17} seem to be 
usable to obtain an approximation gap that depends only on $t$ but not the 
parameter $\ell$.
However, we develop a weaker parameterized inapproximability 
result for the case when 
$t\geq\xi(\ell)=\ell^{o(1)}$ or $t\geq\zeta(\ell)=\omega(1)$ 
(see~\cref{thm:sparse-MCSI} in~\cref{sec:parameterH}), 
and use it to prove \cref{thm:starsparsified}.

%% file: src/Kab.tex
In this section we give a polynomial-time 
$\Oh\bigl((a+b)^{1/a} \cdot \alpha(G)^{1-1/a}\bigr)$-approximation algorithm for $\IS$ on 
$K_{a,b}$-free graphs, where $\alpha(G)$ is the size of a maximum independent 
set in the input graph $G$. The algorithm is a generalization of a known local 
search procedure. Note that it asymptotically matches the approximation factor 
of the $(\frac{d-1}{2}+\delta)$-approximation algorithm for $K_{1,d}$-free 
graphs of Halld{\'{o}}rsson~\cite{DBLP:conf/soda/Halldorsson95} by setting $a=1$ and 
$b=d$. We note here that the following theorem was independently discovered by 
Bonnet, Thomass{\'{e}}, Tran, and Watrigant~\cite{DBLP:conf/esa/BonnetTTW20}.

\kabtheorem*
\begin{proof}
The algorithm first computes a maximal independent set $I \subseteq V(G)$ in the 
given graph $G$, which can be done in linear time using a simple greedy 
approach. Since $I$ is maximal, every vertex in $V(G) \setminus I$ has at least 
one neighbor in $I$. Now, we consider the vertices in $V(G) \setminus I$ that 
are neighbors to at most $a-1$ vertices of~$I$, and call this set $V_1$. Let 
$C\subseteq I$ be a set of size $c\in[a-1]$, and let $V_C := \{ v \in V_1 \mid 
N(v)\cap I = C\}$. If the graph induced by  $V_C\cup C$ contains an independent 
set $I'$ of size $|C|+1$, then we can find it in time 
$n^{\Oh(|C|+1)}=n^{\Oh(a)}$. Furthermore, $(I\setminus C)\cup I'$ is an 
independent set, since no vertex of $V_C\cup C$ is adjacent to any vertex of 
$I\setminus C$, and $(I\setminus C)\cup I'$ is larger by one than $I$. Thus the 
algorithm replaces $I\setminus C$ by $I'$ in $I$. The algorithm repeats this 
procedure until the largest independent set in each subgraph induced by a set 
$V_C\cup C$ (defined for the current~$I$) is of size at most $|C|$. At this 
point the algorithm outputs $I$.

Let $k=|I|$ be the size of the output at the end of the algorithm. We claim 
that $\alpha(G)\leq (a-1)k^{a-1}+(b-1) k^a=\Oh\bigl((a+b)k^a\bigr)$ and this would 
prove the theorem, since then $k=\Omega\bigl((\frac{\alpha(G)}{a+b})^{1/a}\bigr)$, which 
implies that $I$ is an $\Oh\bigl((a+b)^{1/a} \cdot \alpha(G)^{1-1/a}\bigr)$-approximation.

To show the claim, first note that the family $\bigl\{V_C\ \mid C \subseteq I \text{ and } |C| \in [a-1]\bigr\}$
is a partition of $V_1$ into at most  $\sum_{c=1}^{a-1}{k\choose c}$ many sets.
For each relevant $C$, no subgraph induced by a set $V_C\cup C$ contains an independent set larger than $|C|$,
and thus if $I^*$ denotes a
maximum independent set of $G$, then $\bigl|(V_C\cup C)\cap I^*\bigr|\leq |C|$.
Thus, 
\[
\bigl|(V_1\cup I)\cap I^*\bigr|\leq \sum_{c=1}^{a-1}c{k\choose c} = \sum_{c=1}^{a-1} k\binom{k-1}{c-1} \leq \sum_{c=1}^{a-1}k^c\leq (a-1)k^{a-1}.  
\]

Now consider the remaining set $V_2: = V(G) \setminus (V_1 \cup I)$, and observe 
that every $v \in V_2$ has at least~$a$ neighbors in $I$ due to the definition 
of~$V_1$. For each $D \subseteq I$ with $|D|=a$, we construct a set $V_{D}$ by 
fixing an arbitrary subset $S(v) \subseteq (N(v) \cap I)$ of size $a$ for every 
$v\in V_2$, and putting $v$ into~$V_D$ if and only if $S(v) = D$. Observe that 
these sets $V_D$ form a partition of $V_2$ of size at most ${k \choose a}$. We 
claim that each $V_D$ induces a subgraph of $G$ for which every independent set 
has size less than~$b$. Assume not, and let $I'$ be an independent set in $V_D$ 
of size~$b$. But then $D\cup I'$ induces a $K_{a,b}$ in $G$, since every vertex 
of $I'\subseteq V_D$ is adjacent to every vertex of~$D\subseteq I$. As this 
contradicts the fact that $G$ is $K_{a,b}$-free, we have $|V_D\cap I^*|\leq 
b-1$, and consequently $|V_2\cap I^*|\leq (b-1){k \choose a}\leq (b-1) k^a$. 
Together with the above bound on the number of vertices of $I^*$ in $V_1\cup I$ 
we get 
\[
\alpha(G)=|I^*|\leq (a-1)k^{a-1}+(b-1) k^a, 
\]
which concludes the proof.
\end{proof}

%% file: src/k_1_d_free_polynomial_inapprox.tex
In this section, we show polynomial time approximation lower bounds for \IS\ on 
$K_{1,d}$-free graphs. 

\starfreepoly*

For that, we reduce from the \MCSI{$t$} problem, and leverage the lower bound by  Laekhanukit~\cite[Theorem~$6$]{laekhanukit2014parameters}.
Let us point out that the original statement of the lower bound by Laekhanukit~\cite{laekhanukit2014parameters} is in terms of the \textsc{Label Cover} problem,
but, as we already mentioned, this is a special case of \MCSI{}.

\begin{theorem}[Laekhanukit~\cite{laekhanukit2014parameters}]
\label{thm:Bundit_polytime}
Let $\Gamma = (G,V_1,\dots,V_\ell,J)$ be an instance of \MCSI{$t$} where $J$ is 
a bipartite graph. Assuming \ZPP~$\neq$~\NP, there exist constants $t^\star$ 
and $c$ such that for any constant $\eps >0$ and any $t\geq t^\star$, there is 
no polynomial time algorithm that can distinguish between the  two cases: 
\begin{compactenum}
 \item {(YES-case)} $\val(\Gamma) \geq 1-\eps$, and 
 \item {(NO-case)} $\val(\Gamma) \leq c\log(t)/t+\eps$.
\end{compactenum}
\end{theorem}

We use a standard reduction from \MCSI{} to \IS, which can be seen as a variant of the so-called \emph{FGLSS-graph}~\cite{DBLP:journals/jacm/FeigeGLSS96}. For instances of 
\MCSI{$t$} with bounded degree $t$ gives the following lemma.

\begin{restatable}{lemma}{starfreereduction}
\label{lemma:K1dfreereduction}
Let $\Gamma = (G,V_1,\dots,V_\ell,J)$ be an instance of \MCSI{$t$}. Given 
$\Gamma$, in polynomial time we can construct 
an instance $G'$ of \IS\  such that
\begin{compactenum}
\item $G'$ does not have $K_{1,d}$ as an induced subgraph for any $d \geq 
2t +2 $,
\item if $\val(\Gamma) \geq \mu$ then $G'$ has an independent set of 
size at least $\mu |E(J)|$, and
\item if  $\val(\Gamma) \leq \nu$ then every independent set of 
$G'$ has size at most $\nu |E(J)|$.
 \end{compactenum}
\end{restatable}

\input{src/proofLemmaK1dfree}

Now we are ready to prove \cref{thm:k1dfree}.
\begin{proof}[Proof of \cref{thm:k1dfree}]
Assume there was a polynomial time algorithm $\mathcal{A}$ to approximate the 
\IS problem within a factor $\frac{1-\eps}{c\log(t)/t+\eps}$ for some $\eps>0$ 
in $K_{1,d}$-free graphs, where $t=\lfloor\frac{d}{2}-1\rfloor\geq t^\star$, and 
$c$ is the constant given by~\cref{thm:Bundit_polytime}. Given an instance 
$\Gamma = (G,V_1,\dots,V_\ell,J)$ of \MCSI{$t$} and $\eps$, we can reduce it to 
an instance of \IS in $K_{1,d}$-free graphs in polynomial time by using the 
reduction of \cref{lemma:K1dfreereduction}. Now, setting $\mu = 1 - \eps$ and 
$\nu = (c\log(t)/t)+\eps$ in the statement of~\cref{lemma:K1dfreereduction}, 
this gives that given an instance $\Gamma$ of \MCSI{$t$} and $\eps$, we can now 
use~$\mathcal{A}$ to differentiate between the YES- and NO-cases 
of~\cref{thm:Bundit_polytime} in polynomial time, which would mean that 
\ZPP~=~\NP. As $\frac{1-\eps}{c\log(t)/t+\eps}=\Oh(d/\log d)$, this implies 
\cref{thm:k1dfree}, where $d^\star$ is the constant for which 
$\lfloor\frac{d^\star}{2}-1\rfloor= t^\star$.
\end{proof}

%% file: src/proofLemmaK1dfree.tex
\begin{proof}
We first describe the construction of $G'$ given $\Gamma = 
(G,V_1,\dots,V_\ell,J)$, where we denote by $E_{ij}$ the edge set between $V_i$ 
and $V_j$ for each edge $ij\in E(J)$. 
The graph $G'$ has a vertex $v_e$ for each edge $e$ of~$G$, an edge between 
$v_e$ and $v_f$ if $e,f\in E_{ij}$ for some $ij\in E(J)$, and an edge between 
$v_e$ and $v_f$ if $e\in E_{ij}$ and $f\in E_{ij'}$ and $e$ and $f$ do not share 
a vertex in $G$ for some three vertices~$i,j,j'\in[\ell]$ of $J$ such that 
$ij\in E(J)$ and $ij'\in E(J)$. Note that the vertex set $V'_{ij}=\{v_e\in 
V(G')\mid e\in E_{ij}\}$ induces a clique in $G'$.
This finishes the construction of $G'$. 
See Figure~\ref{fig:K_1d_construction} for better understanding of the construction.
\begin{figure}[h]
 \centering
 \includegraphics{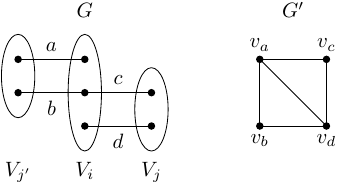}
 \caption{Example of the construction of the graph $G'$ from $G$ for $J$ being a path on 3 vertices.}
 \label{fig:K_1d_construction}
\end{figure}

To see the first part of the lemma, for the sake of contradiction, let us 
suppose~$G'$ has a $K_{1,d}$ as an induced subgraph for $d \geq 2t +2$. We know 
that for any $e\in E(J)$ the vertices in $V'_{e}$ form a clique in $G'$, so the 
star $K_{1,d}$ can intersect with a fixed $V'_{e}$ in at most two vertices 
of which one must be the center vertex of $K_{1,d}$ with degree $d$. As 
$K_{1,d}$ has $d+1$ vertices, this means there are (at least) $d$ distinct 
vertex sets $V'_{e_1}, V'_{e_2}, \ldots, V'_{e_d}$ of $G$ that intersect the 
$K_{1,d}$ for some edges $e_1,e_2,\ldots,e_d\in E(J)$. Without loss of 
generality, let the center vertex of the $K_{1,d}$ come from $V'_{e_1}$. Note 
that the $K_{1,d}$ has an edge between a vertex from $V'_{e_1}$ and a vertex 
from $V'_{e_z}$ for each $z \in \{2,\ldots,d\}$. Hence if $e_1 = jj'$, we have 
that either $j \in e_z$ or $j' \in e_z$ for every $z \in [d]$ by the 
construction of $G'$. This means that either $j$ or $j'$ has at least $(d-1)/2$ 
neighbours in $J$. That is, the maximum degree of $J$ is at least $(d-1)/2$.
As $d \geq 2t+2$, we obtain that the maximum degree of $J$ is more than $t$,
which is a contradiction with the definition of \MCSI{$t$}.

Now, to see the second claim of the lemma, first we need to show that if 
$\val(\Gamma) \geq \mu$, then $G'$ has an independent set of size at least 
$\mu|E(J)|$. To see that, let $\phi: V(J) \rightarrow V(G)$ be a mapping that 
satisfies at least a $\mu$-fraction of the edges of $E(J)$. We claim that $S= 
\{v_{uw}\in V'_{ij} \mid ij\in E(J), \phi(i) = u, \phi(j) = w\}$ is an 
independent set of size at least $\mu|E(J)|$ in~$G'$. Since $\phi$ satisfies at 
least $\mu$-fraction of edges, $S$ has size at least $\mu|E(J)|$. So all we need 
to show is that $S$ is indeed an independent set. Suppose it was not the case, 
i.e., there exist $v_e, v_f \in S$ that are adjacent in~$G'$. By construction of 
$G'$ there can be an edge between $v_e$ and $v_f$ only if $e\in E_{ij}$ and 
$f\in E_{ij'}$ where possibly $j=j'$. Note that $\phi(i)=u\in V_i$ is a common 
endpoint of both $e$ and $f$. If indeed $j=j'$, then $\phi(j)=w\in V_j$ 
is also a common endpoint of both $e$ and $f$, so that $e=f$, i.e., $v_e$ and $v_f$ 
are not distinct. Hence it must be that $j\neq j'$. But in this case, the 
construction of $G'$ implies that $e$ and $f$ do not share a vertex, which 
contradicts the fact that they have $u$ as a common endpoint.

For the third part of the lemma, we prove the contrapositive: we claim that if 
$G'$ has an independent set $S$ of size $k \geq \nu|E(J)|$, then there exists an 
assignment $\phi: V(J) \rightarrow V(G)$ satisfying at least $k$ edges in 
$\Gamma$. To see that, first observe that the set $S$ can contain at most one 
vertex from $V'_e$ as any two vertices in $V'_e$ are adjacent. Let $E_S := \{ e 
\in E(J) \mid S \cap V'_e \neq \emptyset\}$, for which we then have $|E_S| = 
|S|$. We claim that all the edges in $E_S$ can be satisfied by an assignment 
$\phi$ defined as follows. For $ij \in E_S$, let $S \cap V'_{ij} = \{v_{uw}\}$. 
Then we set $\phi(i)=u$ and $\phi(j)=w$. We need to show that the function 
$\phi$ is well-defined. Suppose some vertex $i \in V(J)$ gets mapped to more 
than one vertex of $V(G)$ by $\phi$. This must mean that there exist two edges 
in $G$ that contain one endpoint in $V_i$ and are in~$E_S$. But this would mean 
that the two vertices in $S$ corresponding to these two edges in $E_S$ are 
adjacent due to the construction of $G'$. This is a contradiction to~$S$ being 
an independent set. Also, $\phi(i)\phi(j)\in E(G)$ for all $ij \in E_S$, since 
for each $v_{uw}$ we have $uw \in E(G)$, and we have set $\phi(i) = u$ and 
$\phi(j)= w$. This concludes the proof. 
\end{proof}

%% file: src/fpt_apx_hardness2.tex
In this section we prove \cref{thm:noFPAS} and \cref{thm:NoConstantFptApprox}, that follows from \cref{thm:noFPAS} using a gap amplification.
Thus, we first prove \cref{thm:noFPAS}.
Let us define an auxiliary family of classes of graphs:
for integers $4 \leq a \leq b$ and $c \geq 3$, let $\cH([a,b],c)$ be a family of graph consists of $K_{1,c}$ and cycles $C_p$ for all $p \in [a,b]$.
Further, let $\cC([a,b],c)$ be a class of $\cH([a,b],c)$-free graphs.
Let ${\cal T}(b')$ be the class of trees with two vertices of degree at least 3 at distance at most $b'$.
Let $\cC^*([a,b],c) \subseteq \cC([a,b],c)$ be the set of those $G \in \cC([a,b],c)$, which are are also ${\cal T}(\lceil\frac{b - 1}{2}\rceil)$-free, i.e., $\cC^*([a,b],c)$ consists of $\cH([a,b],c) \cup {\cal T}(\lceil\frac{b - 1}{2}\rceil)$-free graphs.
Actually, we will prove the following theorem, which implies \cref{thm:noFPAS}.

\begin{restatable}{theorem}{noFPASclasses}
\label{thm:noFPAS-classes}
Let $z \geq 5$ be a constant.
The following lower bounds hold for the \IS problem on graphs $G \in 
\cC^*([4,z],5) \cup \cC([5,z],4)$ with $n$ vertices.
\begin{compactenum}
\item For any computable function $f$, there is no 
$f(k)\cdot n^{o(k/\log k)}$-time algorithm that determines if $\alpha(G) \geq k$, unless the ETH fails.
\item There exists a constant $\gamma > 
0$, such that for any computable function $f$, there is no $f(k)\cdot 
n^{\Oh(1)}$-time algorithm that can distinguish between the  two 
cases: $\alpha(G) \geq k$, or $\alpha(G) < (1-\gamma)\cdot k$, unless the 
deterministic Gap-ETH fails.
\label{it:NoFptApprox-classes}
\item There exists a constant $\gamma > 0$, 
such that for any computable function $f$, there is no $f(k)\cdot 
n^{o(\sqrt{k})}$-time algorithm that can distinguish between the  two 
cases: $\alpha(G) \geq k$, or $\alpha(G) < (1-\gamma)\cdot k$, unless the 
randomized Gap-ETH fails.
\end{compactenum}
\end{restatable}

The proof of \cref{thm:noFPAS-classes} consists of two steps: first we will 
prove it for graphs in $\cC^*([4,z],5)$, and then for graphs in $ \cC([5,z],4)$. In 
both proofs we will reduce from the \MCSI{$3$} problem.
Let $\Gamma = (G,V_1,\dots,V_\ell,J)$ be an instance of \MCSI{$3$}.
For $ij \in E(J)$, by $E_{ij}=E_{ji}$ we denote the set of edges between $V_i$ and $V_j$.
Note that we may assume that $J$ has no isolated vertices, each $V_i$ is an independent set, and $E_{ij}\neq\emptyset$ if and only if $ij \in E(J)$.

Lokshtanov et al.~\cite{Lokshtanov17} gave the following hardness result (the 
first statement actually follows from Marx~\cite{Marx10} and Marx, 
Pilipczuk~\cite{MarxP15}). We note that Lokshtanov et al.~\cite{Lokshtanov17} 
conditioned their result on the \emph{Parameterized Inapproximability 
Hypothesis} (PIH) and \Wone~$\neq$~FPT. Here we use stronger assumptions, i.e., 
the deterministic and 
randomized Gap-ETH, which are more standard in the area of 
parameterized approximation. The reduction in~\cite{Lokshtanov17} yields the 
following theorem, when starting from 
\cite{DBLP:conf/innovations/DinurM18,DBLP:journals/corr/abs-1805-03867} and 
\cite{DBLP:conf/soda/Manurangsi20}, respectively (see also 
\cite[Corollary~7.9]{chitnis2017parameterized}).

\begin{theorem}[Lokshtanov et al.~\cite{Lokshtanov17}]
\label{thm:MCSIHardness}
Consider an arbitrary instance $\Gamma = (G,V_1,\dots,V_\ell,J)$ of \MCSI{$3$} with size $n$.
\begin{compactenum}
\item Assuming the ETH, for any computable function $f$, there is no 
$f(\ell)\cdot n^{o(\ell/\log \ell)}$ time algorithm that solves $\Gamma$.
\item Assuming the deterministic Gap-ETH there exists a constant $\gamma > 
0$, such that for any computable function $f$, there is no $f(\ell)\cdot 
n^{\Oh(1)}$ time algorithm that can distinguish between the  two 
cases: \emph{(YES-case)} $\val(\Gamma) = 1$, and \emph{(NO-case)} $\val(\Gamma) 
< 1 - \gamma$.
\item Assuming the randomized Gap-ETH there exists a constant $\gamma > 0$, 
such that for any computable function $f$, there is no $f(\ell)\cdot 
n^{o(\sqrt{\ell})}$ time algorithm that can distinguish between the  two 
cases: \emph{(YES-case)} $\val(\Gamma) = 1$, and \emph{(NO-case)} $\val(\Gamma) 
< 1 - \gamma$. \label{it:xp-gap}
\end{compactenum}
\end{theorem}

\subsection{Hardness for $(C_4,C_5\ldots,C_z,K_{1,5},{\cal 
T}(\lceil\frac{z-1}{2}\rceil))$-free Graphs} \label{sec:hardness-k15}

First, let us show \cref{thm:noFPAS-classes} for $\cC^*([4,z],5)$, i.e., for  
$(C_4,C_5\ldots,C_z,K_{1,5}, {\cal T}(s))$-free graphs for $s = \lceil\frac{z-1}{2}\rceil$.
Let $\Gamma = (G,V_1,\dots,V_\ell,J)$ be an instance of \MCSI{3}.
We aim to build an instance $(G',k)$ of {\sc Independent Set}, such that the graph $G' \in \cC^*([4,z],5)$.

For each $ij \in E(J)$, we introduce a clique $C_{ij}$ of size $|E_{ij}|$,
whose every vertex \emph{represents} a different edge from $E_{ij}$. 
The cliques constructed at this step will be called \emph{primary cliques}, note that their number is $|E(J)|$.
Choosing a vertex $v$ from $C_{ij}$ to an independent set of $G'$ will correspond to mapping $i$ and $j$ to the appropriate endvertices of the edge from $E_{ij}$, corresponding to $v$.

Now we need to ensure that the choices in primary cliques corresponding to edges of $G$ are consistent.
Consider $i \in V(J)$ and suppose it has three neighbors $j_1,j_2,j_3$ (the cases if $i$ has fewer neighbors are dealt with analogously).
We will connect the cliques $C_{ij_1},C_{ij_2},C_{ij_3}$ using a gadget called a \emph{vertex-cycle}, whose construction we describe below.
For each $a \in \{1,2,3\}$, we introduce $s$ copies of $C_{ij_a}$ and denote them by $D^1_{ij_a},D^2_{ij_a},\ldots,D^s_{ij_a}$, respectively. 
Let us call these copies \emph{secondary cliques}.
The vertices of secondary cliques represent the edges from $E_{ij_a}$ analogously as the ones of $C_{ij_a}$. 
We call primary and secondary cliques as \emph{base cliques}.
We connect the base cliques corresponding to the vertex $i \in V(J)$ into \emph{vertex-cycle} ${\cal C}_i$.
Imagine that secondary cliques, along with primary cliques $C_{ij_1},C_{ij_2},C_{ij_3}$, are arranged in a cycle-like fashion, as follows:
\[
C_{ij_1},D_{ij_1}^1,D_{ij_1}^2,\ldots,D_{ij_1}^s,C_{ij_2},D_{ij_2}^1,D_{ij_2}^2,\ldots,D_{ij_2}^s,C_{ij_3},D_{ij_3}^1,D_{ij_3}^2,\ldots,D_{ij_3}^s,C_{ij_1}.
\]
This cyclic ordering of cliques constitutes the vertex-cycle, let us point out that we treat this cycle as a directed one.
As we describe below we put some edges between two base cliques $B_1$ and $B_2$ 
only if they belong to some vertex-cycle~${\cal C}_i$.
See Figure~\ref{fig:K_15_construction} for an example of how we connect base cliques.
\begin{figure}[b]
 \centering
 \includegraphics[scale=1]{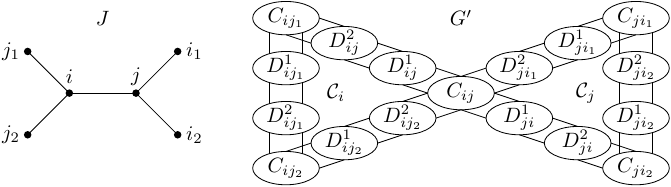}
 \caption{A part of the construction of $G'$ for $s = 2$. Cliques $C_{ab}$ 
representing edge sets $E_{ab} \subseteq E(G)$ are connected through secondary 
cliques $D^p_{ab}$.}
 \label{fig:K_15_construction}
\end{figure}

Now, we describe how we connect the consecutive cliques in ${\cal C}_i$.
Recall that each vertex $v$ of each clique represents exactly one edge $uw$ of $G$, whose exactly one vertex, say $u$, is in $V_i$. 
We extend the notion of representing and say that $v$ \emph{represents} $u$, and denote it by $\repvert{i}(v)=u$.

Let us fix an arbitrary ordering $\prec_i$ on $V_i$. 
Now, consider two consecutive cliques of the vertex-cycle. 
Let $v$ be a vertex of the first clique and $v'$ be a vertex from the second clique, and let $u$ and $u'$ be the vertices of $V_i$ represented by $v$ and $v'$, respectively. 
The edge $vv'$ exists in $G'$ if and only if $u \prec_i u'$.
See Figure~\ref{fig:half_graph_rel} how we connect two consecutive base cliques in a vertex-cycle.
This finishes the construction of ${\cal C}_i$.
\begin{figure}[b]
 \centering
 \includegraphics[scale=0.82]{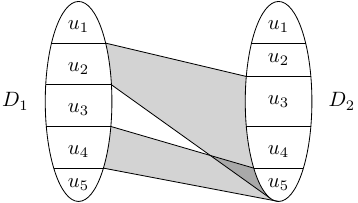}
 \caption{Example of edges between two consecutive cliques $B_1$ and $B_2$ in a vertex-cycle ${\cal C}_i$, where $V_i = \{u_1,\dots,u_5\}$.
 Each region marked with $u_b$ in $B_1$ and $B_2$ contains vertices corresponding to edges of $E_{ij} \subseteq E(G)$ incident to $u_b$.
 Thus, the vertices from the region $u_b$ in $B_1$ are connected to the vertices of regions $u_{b'}$ in $B_2$ for all $b' > b$.
 For simplicity, we depicted only edges incident to regions $b_2$ and $b_4$ in $B_1$ .
}
 \label{fig:half_graph_rel}
\end{figure}

We introduce a vertex-cycle ${\cal C}_i$ for every vertex $i$ of $J$, note that each primary clique $C_{ij}$ is in exactly two vertex-cycles: ${\cal C}_i$ and ${\cal C}_j$. 
The number of all base cliques is
\[
k:=\underbrace{|E(J)|}_{\substack{\text{primary}\\\text{cliques}}} + \underbrace{\sum\nolimits_{i \in V(J)} \deg_J(i) \cdot s }_{\text{secondary cliques}} = |E(J)| \cdot \Bigl(1 + \frac{s}{2}\Bigr) \leq \frac{3\ell}{2} \cdot \Bigl(1 + \frac{s}{2}\Bigr) = \Oh(\ell).
\]
This concludes the construction of $(G',k)$. 
Since $V(G')$ is partitioned into $k$ base cliques, $k$ is an upper bound on the size of any independent set in $G'$, and a solution of size $k$ contains exactly one vertex from each base clique.

We claim that the graph $G'$ is in the class $\cC^*([4,z],5)$. 
Moreover, if $\val(\Gamma) = 1$, then the graph $G'$ has an independent set of size $k$ and
if the graph $G'$ has an independent set of size at least $(1 - \gamma')\cdot k$ 
for $\gamma' = \frac{\gamma}{6 + 3s}$, then $\text{val}(\Gamma) \geq 1 
- \gamma$.
By Theorem~\ref{thm:MCSIHardness}, we conclude Theorem~\ref{thm:noFPAS-classes} holds for the class $\cC^*([4,z],5)$.

Now, we will prove our claims about $G'$.
For two distinct base cliques $B_1,B_2$, by $E(B_1,B_2)$ we denote the set of edges with one endvertex in $B_1$ and another in $B_2$. We say that $B_1,B_2$ are \emph{adjacent} if $E(B_1,B_2)\neq\emptyset$.

\begin{claimm}\label{lem:NoInducedMatching}
Let $B_1, B_2$ be two distinct base cliques in $G'$.
Then the size of a maximum induced matching in the graph induced by $E(B_1,B_2)$ is at most 1.
\end{claimm}
\begin{proof}
If $E(B_1,B_2)$ is empty, then the lemma holds trivially.
Consider two disjoint edges $e=v_1v_2$ and $e'=v'_1v'_2$ in $E(B_1,B_2)$, where $v_1,v'_1 \in B_1$ and $v_2, v'_2 \in B_2$.
We prove that there is an edge $f \in E(B_1,B_2)$ such that $f$ intersect both $e$ and $e'$.
  
By construction, $B_1$ and $B_2$ are consecutive cliques in a vertex-cycle ${\cal C}_i$ for some $i \in V(J)$.
Assume that $B_2$ is the successor of $B_1$ on this cycle. 
Recall that each $v \in \{v_1,v'_1,v_2,v'_2\}$ represents some vertex $\repvert{i}(v) \in V_i$. 
Since $v_1v_2,v'_1v'_2 \in E(G')$, we observe that $\repvert{i}(v_1) \prec_i \repvert{i}(v_2)$ and $\repvert{i}(v'_1) \prec_i \repvert{i}(v'_2)$. 
Thus, at least one of the following holds $\repvert{i}(v_1) \prec_i \repvert{i}(v'_2)$ or $\repvert{i}(v'_1) \prec_i \repvert{i}(v_2)$.
Therefore, at least one of the edges $v_1v'_2$ or $v'_1v_2$ exists in~$G'$.
\end{proof}

\begin{claimm}
\label{lem:Cycles4}
 The graph $G'$ is $(C_4,\dots,C_z)$-free.
\end{claimm}
\begin{proof}
For contradiction, suppose that there exists an induced cycle $K$ in $G'$ with consecutive vertices $(v_1,v_2,\ldots,v_p)$, where $p \in [4,z]$.
Note that two consecutive vertices of $K$ might be in the same base clique, or two adjacent base cliques.
Furthermore, no non-consecutive vertices of $K$ may be in one base clique.

Note that each vertex-cycle in $G'$ has at least $2s+2 > z$ base cliques.
Moreover, if $K$ contains vertices of more than on vertex-cycle,
then it has to contains a vertices of at least 4 primary cliques. 
Thus, the the length of $K$ would be larger than $4s + 4 > z$.
Therefore, we conclude that $K$ cannot intersect more than two base cliques. 
It cannot intersect one base clique, as $p > 3$, so suppose that $K$ intersects exactly two base cliques $B_1$ and $B_2$. 
Observe that this means that $p=4$ and $v_1,v_2 \in B_1$, while $v_3,v_4 \in B_2$.
However, by \cref{lem:NoInducedMatching}, we observe that either $v_1$ and $v_3$, or $v_2$ and $v_4$,  are adjacent in $G'$, so $K$ is not induced.
\end{proof}

\begin{claimm}
\label{lem:Star5}
 The graph $G'$ is $K_{1,5}$-free.
\end{claimm}
\begin{proof}
By contradiction suppose that the set  $\{v,v_1,v_2,v_3,v_4,v_5\} \subseteq V(G')$ induces a copy of $K_{1,5}$ in $G'$ with $v$ being the central vertex. 
Let $B$ be the base clique containing $v$.
Since each of $v_1,v_2,v_3,v_4,v_5$ must be in a different base clique and $B$ is adjacent to at most four other base cliques, we conclude that one of $v_a$'s, say $v_5$, belongs to $B$.
For $a \in [4]$, let $B_i$ be the base clique containing $u_a$. 
Furthermore, note that $B$ must be a primary clique, say $B=C_{ij}$,
since only those ones are adjacent to four base cliques. 
Therefore two of $B_a$'s, say $B_1$ and $B_2$, must belong to the vertex-cycle ${\cal C}_i$.
Let $B_1$ precede $B$, and $B_2$ succeed $B$ on this cycle.
Consider the vertices $\repvert{i}(v),\repvert{i}(v_1),\repvert{i}(v_2),\repvert{i}(v_5)$ and recall that since $v$ is adjacent to $v_1,v_2$, we have $\repvert{i}(v_1) \prec_i \repvert{i}(v) \prec_i \repvert{i}(v_2)$.
However, $v_5$ is non-adjacent to $v_1,v_2$, which means that $\repvert{i}(v_2) \prec_i \repvert{i}(v_5) \prec_i \repvert{i}(v_1)$, which is a contradiction, since $\prec_i$ is transitive.
\end{proof}

\begin{claimm}
\label{clm:Trees}
Let $T \in {\cal T}(s)$.
Then, the graph $G'$ is $T$-free.
\end{claimm}
\begin{proof}
Suppose that $G'$ contains $T$ as an induced subgraph.
Let $v,v' \in V(T)$ such that $\deg_T(v), \deg_T(v') \geq 3$ and $\text{dist}_T(v,v') \leq s$.
Note that any two primary cliques are at distance at least $s + 1$.
Thus, $v$ and $v'$ can not be both in primary cliques.
Without loss of generality, let $v$ be in a secondary clique $D$ of a vertex-cycle ${\cal C}_i$.
There are only two base cliques $B_1$ and $B_2$ adjacent to the secondary clique $D$.
Let $v_1,v_2$ and $v_3$ be distinct neighbors of $v$ in $T$.
Since $v_1,v_2$, and $v_3$ form an independent set in $T$, they have to be in distinct base cliques in $G$.
Thus, we can suppose $v_1 \in V(B_1), v_2 \in V(B_2)$ and $v_3 \in V(D)$.
However, by the same argument as in proof of~\cref{lem:Star5} these four vertices $v,v_1,v_2$, and $v_3$ cannot exist.
\end{proof}

\begin{claimm}
\label{thm:Reduction1}
If $\val(\Gamma) = 1$, then the graph $G'$ has an independent set of size $k$.
\end{claimm}
\begin{proof}
Let $\phi$ be a solution of $\Gamma$ of value 1, i.e., for each $ij \in E(J)$ holds that $\phi(i)\phi(j)$ is an edge of $G$.
We will find an independent set $I$ in $G'$ of size $k$.
For each $ij \in E(J)$ we add to the set $I$ the vertex from the primary clique $C_{ij}$ which represents  the edge $\phi(i)\phi(j)$.
Thus, we pick one vertex from each primary clique.
Recall that each secondary clique $D$ is a copy of some primary clique $C$.
If we pick a vertex $v$ from $C$, then we add to $I$ also a copy of $v$ from $D$.
Thus, we add one vertex from each base clique to the set $I$ and therefore $|I| = k$.
 
We claim that $I$ is independent.
Suppose there exist $v,v' \in I$ such that $vv' \in E(G')$.
Let $v \in V(B_1)$ and $v' \in V(B_2)$ for some base cliques $B_1$ and $B_2$.
First, suppose that $B_1$ and $B_2$ are copies of the same primary clique $C_{ij}$ (or one of them is the primary clique itself and the second one is the copy)\footnote{The possibilities for $\{B_1,B_2\}$ are: $\{C_{ij}, D^1_{ij}\}$ or $\{D^p_{ij}, D^{p+1}_{ij}\}$ for $p < s$.}.
Thus, the vertices $v$ and $v'$ represent the same edge in $E_{ij}$ and by construction,
vertices in primary and secondary cliques representing the same edge in $E_{ij}$ are not adjacent.
 
Therefore $B_1 = D^s_{ij_1}$ and $B_2 = C_{ij_2}$ (or vice versa) for some edges $ij_1$ and  $ij_2$ in $E(J)$.
Edges between $B_1$ and $B_2$ were added according to the ordering $\prec_i$ of vertices in $V_i$.
Note that the vertices $v$ and $v'$ represent edges $\phi(i)\phi(j_1)$ and $\phi(i)\phi(j_2)$.
Thus, $\repvert{i}(v) = \phi(i) = \repvert{i}(w)$.
Since $v$ and $v'$ are adjacent in $G'$, it holds that $\repvert{i}(v) \prec_i \repvert{i}(v')$ by construction,
which is a contradiction with $\repvert{i}(v) = \repvert{i}(v')$.
Therefore, $I$ is an independent set.
\end{proof}

\begin{claimm}
\label{thm:Reduction2}
Let $\gamma > 0$.
If the graph $G'$ has an independent set of size at least $(1 - \gamma')\cdot k$ for $\gamma' = \frac{\gamma}{6 + 3s}$, then $\text{val}(\Gamma) \geq 1 - \gamma$.
\end{claimm}
\begin{proof}
 Let 
\begin{compactitem}
\item $I$ be a maximum independent set of $G'$ of size at least $(1 - \gamma')\cdot k$,
\item $i$ be a vertex of $J$, and suppose its degree is 3 (the case of vertices of smaller degree is treated analogously),
\item $j_1,j_2,j_3$ be the neighbors of $i$ in $J$,
\item $I_i$ be an intersection of $I$ and vertices of cliques in ${\cal C}_i$.  
 \end{compactitem}
Suppose that $|I_i| = 3s + 3$, i.e., $I$ intersects each clique in ${\cal C}_i$.
Let $v_1,v_2,v_3$ be vertices of intersections of $I$ and $C_{ij_1}$, $C_{ij_2}$, and $C_{ij_3}$, respectively. 
We claim that $\repvert{i}(v_1) =  \repvert{i}(v_2) = \repvert{i}(v_3)$.

Denote the consecutive cliques of ${\cal C}_i$ by $B_1,B_2,\ldots,B_{3s+3}$.
Recall that two cliques in ${\cal C}_i$ are adjacent if and only if they are consecutive. 
 For $p \in [3s+3]$ let $v'_p$ be the unique vertex in $I \cap V(B_p)$. 
Define a relation $\succeq_i$ on $V_i$, such that $v \succeq_i v'$ iff $v \not\prec_i v'$.
Since $\prec_i$ is a total order on $V_i$, we have that $v \succeq_i v'$ iff $v' \prec_i v$ or $v=v'$.
Since $v'_1,\dots,v'_{3s+3}$ are pairwise nonadjacent, it holds that $\repvert{i}(v'_1) \succeq_i \repvert{i}(v'_2) \succeq_i \dots \succeq_i  \repvert{i}(v'_{3s+3}) \succeq_i \repvert{i}(v'_1)$ by construction.
This implies that all vertices $v'_p$ represent the same vertex $u \in V_i$, in particular, $r_i(v_1) = r_i(v_2) = r_i(v_3) = u$.

Now, if $|I_i| = 3s + 3$, we define $\phi(i) = u$ (where $u$ is as in the previous paragraph).
If $|I_i| < 3s + 3$ we define $\phi(i)$ arbitrarily. 
Vertices $i' \in V(J)$ of degree 2 are processed similarly, however the size of $I_{i'}$ is compared to value $2s + 2$.
We say that the set $I_i$ is complete if $|I_i| = (s + 1)\cdot \deg(i)$.
Thus, if $I_i$ and $I_j$ are complete, then $\phi(i)\phi(j)$ is an edge of $G$.
 
Let $Q \subseteq V(J)$ be a set of vertices $i$ of $J$ such that $I_i$ is not complete.
Note that a primary clique $C_{ij}$ is in two vertex-cycles of base cliques ${\cal C}_i$ and ${\cal C}_j$ and each secondary clique is in exactly one vertex-cycle of base cliques.
Since there are fewer than $\gamma'\cdot k$ base cliques $B$ such that $I \cap 
B = \emptyset$, the set $Q$ has size less than $2\gamma'\cdot k$.
The vertices in $Q$ are incident to at most $6\gamma'\cdot k$ edges in $J$, and 
all remaining edges of $J$ are satisfied by $\phi$.
 Therefore, 
 \[
  \text{val}(\Gamma) \geq \frac{|E(J)| - 6\gamma' \cdot k}{|E(J)|} = 1 - 
6\gamma' \cdot \Bigl(1 + \frac{s}{2}\Bigr) = 1 - \gamma.\qedhere
 \]
\end{proof}

This completes the proof of \cref{thm:noFPAS-classes} in this case.

\subsection{Hardness for $(C_5\ldots,C_z,K_{1,4})$-free Graphs}

In this section we show \cref{thm:noFPAS-classes} for $\cC([5,z],4)$, i.e., for $(C_5\ldots,C_z,K_{1,4})$-free graphs. 
The proof is similar to the case of $\cC^*([4,z],5)$.
Let $\Gamma = (G,V_1,\dots,V_\ell,J)$ be an instance of \MCSI{3}, we will create an instance $(G',k)$ of \IS, where $G' \in \cC(5,z,4)$.
Consider an edge $ij$ of $J$.
We introduce four \emph{primary cliques} $C^1_{ij},C^2_{ij},C^3_{ij},C^4_{ij}$, each of size $|E_{ij}|$. For each $q \in [4]$, each vertex $v$ of $C^q_{ij}$ \emph{represents} one edge in $E_{ij}$, denote this edge by $\repedge(v)$.

For each $q \in [4]$, we create $s:=\lceil (z-3)/4\rceil$ copies of $C^q_{ij}$, denoted by $D^{q,1}_{ij},\ldots,D^{q,s}_{ij}$.
Each vertex of a copy represents the same edge as the corresponding vertex in $C^q_{ij}$.
The cliques created in this step will be called \emph{cycle cliques}. Again, we imagine that the primary and cycle cliques are arranged in a cyclic way and constitute the \emph{edge-cycle} corresponding to $ij$:
\[
C^1_{ij},D^{1,1}_{ij},\ldots,D^{1,s}_{ij},C^2_{ij},D^{2,1}_{ij},\ldots,D^{2,s}_{ij},C^3_{ij},D^{3,1}_{ij},\ldots,D^{3,s}_{ij},C^4_{ij},D^{4,1}_{ij},\ldots,D^{4,s}_{ij},C^1_{ij}.
\]
Note that all cliques in the edge-cycle are identical. We fix some arbitrary ordering $\prec_{ij}$ on $E_{ij}$, For each two consecutive cliques $B_1$ and $B_2$ of the edge-cycle, where $B_1$ precedes $B_2$, and for any vertex $v_1$ from $B_1$ and any vertex $v_2$ from $B_2$, we make $v_1v_2$ adjacent in $G'$ if and only if $\repedge(v_1) \prec_{ij} \repedge(v_2)$.

After repeating the previous step for every edge $ij$ of $J$, we arrive at the point that $G'$ consists of separate edge-cycles, one for each edge of $J$. Since $J$ has maximum degree 3, each edge of $J$ intersects at most 4 other edges.
So for each pair of intersecting edges $ij$ and $ij'$ we can assign a pair of primary cliques, one in the edge-cycle corresponding to $ij$, and the other one in the edge-cycle corresponding to $ij'$, so that no primary clique is assigned twice.

Consider two edges of $J$, that share a vertex, say edges $ij$ and $ij'$, and suppose the primary cliques chosen in the last step are $C^p_{ij}$ and $C^q_{ij'}$.
We need to provide some connection between these cliques, to make the choices for edges $ij$ and $ij'$ consistent.
Let us arbitrarily choose one of cliques $C^p_{ij}$ and $C^q_{ij'}$, say $C^p_{ij}$, and create $s$ copies of it, denote these cliques by $F^1_{ijj'},F^2_{ijj'},\ldots,F^s_{ijj'}$ (again, the represented edges are inherited from the primary clique). We call these cliques \emph{equality cliques}. We build an \emph{equality gadget} by arranging these cliques in a sequence as follows:
\[
C^p_{ij},F^1_{ijj'},F^2_{ijj'},\ldots,F^s_{ijj'},C^q_{ij'}.
\]
Consider two consecutive cliques $B_1$ and $B_2$ of this sequence, except for the last pair.
These cliques are identical. Between them we add edges that form an antimatching, i.e., for a vertex $v_1$ of $B_1$ and a vertex $v_2$ of $B_2$, we add an edge $v_1v_2$ if and only if $\repedge(v_1) \neq \repedge(v_2)$.
Finally, for a vertex $v_1$ of $F^s_{ijj'}$ and a vertex $v_2$ of $C^q_{ij}$, we add an edge $v_1v_2$ if and only if $\repedge(v_1) \cap \repedge(v_2) \neq \emptyset$, i.e., edges represented by these vertices contain different vertices from $V_i$.

This completes the construction of $G'$. By \emph{base cliques} we mean primary cliques, cycle cliques, and equality cliques.
Let $k$ be the number of all base cliques, i.e.,
\[
k:=\underbrace{4|E(J)|}_{\substack{\text{primary}\\\text{cliques}}} + \underbrace{4s|E(J)|}_{\substack{\text{cycle}\\\text{cliques}}} + \underbrace{\sum_{i \in V(J)} \binom{\deg_J(i)}{2} \cdot s}_{\text{equality cliques}} = \Oh(\ell).
\]
Let us upper-bound $k$. If $\ell_2$ and $\ell_3$ are, respectively, the numbers of vertices of $J$ with degree 2 and 3, then we obtain
\begin{equation}\label{eq:k-upperbound}
k = 4|E(J)|(s+1) + s(\ell_2 + 3\ell_3) \leq \frac{9s}{2} \cdot |E(J)| + 4 \leq 5s \cdot |E(J)|.
\end{equation}

The following claim is proven in an analogous way to \cref{lem:Cycles4}, note 
that this time we might obtain induced copies of $C_4$, where two vertices are 
in an equality clique, and the other two are in a different base clique in the 
same equality gadget (either an equality clique or a primary clique).

\begin{claimm}
\label{lem:Cycles5}
 The graph $G'$ is $(C_5,\dots,C_z)$-free.
\end{claimm}

The next claim is in turn analogous to  \cref{lem:Star5}.

\begin{claimm}
\label{lem:Star4}
 The graph $G'$ is $K_{1,4}$-free.
\end{claimm}
\begin{proof}
Observe that each clique is adjacent to at most three other cliques, and the 
only cliques adjacent to three other cliques are primary cliques. So if we hope 
to find an induced $K_{1,4}$, the center and one leaf must be in a primary 
clique, say $C^q_{ij}$, and other three leaves are in distinct base cliques 
adjacent to $C^q_{ij}$. However, two of cliques adjacent to $C^q_{ij}$ must 
belong to the same edge-cycle (and the third one is an equality clique). 
Similarly as in the proof of \cref{lem:Star5}, we observe that the leaf that 
belongs to $C^q_{ij}$ must be adjacent to at least one of the remaining leaves.
\end{proof}

The following claims are analogous to the corresponding claims in \cref{sec:hardness-k15}.
Therefore we provide only sketches of proofs.

\begin{claimm}
\label{thm:ReductionK14-1}
 If $\val(\Gamma) = 1$, then the graph $G'$ has an independent set of size $k$.
\end{claimm}
\begin{proof}
Consider a solution $\phi$ of $\Gamma$ of value 1. 
Therefore, for each $ij \in E(J)$, the pair $\phi(i)\phi(j)$ is an edge of $G$. 
Note that this edge is represented by some $v$ in each primary clique $C^q_{ij}$. 
We select those vertices to the set $I$. 
Recall that each remaining clique $B$ (i.e., a cycle clique or an equality clique), is a copy of some primary clique $C$. For each such clique $B$ we include to $I$ the vertex, which is a copy of the selected vertex in $C$.

By an argument analogous to the one in the proof of Claim~\ref{thm:Reduction1} we observe that the selected vertices belonging to one edge-cycle are pairwise non-adjacent.
Furthermore, note that the edges between adjacent cliques in an equality gadget are defined in a way, so that all selected vertices from cliques in this gadget are pairwise non-adjacent.
Thus, the $I$ is an independent set of size $k$.
\end{proof}

\begin{claimm}
\label{thm:ReductionK14-2}
 Let $\gamma > 0$.
 If the graph $G'$ has an independent set of size at least $(1 - \gamma')\cdot 
k$ for $\gamma' = \frac{\gamma}{45s}$, then $\text{val}(\Gamma) \geq 1 
- \gamma$.
\end{claimm}
\begin{proof}
Consider an independent set $I$ in $G$ of size at least $(1-\gamma')\cdot k$, and a vertex $i \in V(J)$. Suppose that $\deg (i) = 3$ and the neighbors of $i$ in $J$ are $j_1,j_2,j_3$ (if the degree of $i$ is smaller, the reasoning is analogous).

Let $\mathcal{S}^i$ be the union of all base cliques corresponding to $i$, i.e., 
\begin{enumerate}
 \item belonging to edge-cycles corresponding to $ij_1,ij_2,ij_3$, and
 \item belonging to equality gadgets between these edge-cycles.
\end{enumerate}
Note that the number of cliques in $\mathcal{S}^i$ is $3 \cdot 4(s+1) + 3 \cdot s = 15s+12$, and let $I_i$ be the intersection of $I$ with the vertices of $\mathcal{S}^i$.
Suppose that the size of $I_i$ is $15s+12$, i.e., we selected a vertex from each base clique in $\mathcal{S}^i$ -- we call such $I_i$ complete. 
By the reasoning analogous to Claim~\ref{thm:Reduction2}, we observe that for each of three edge-cycles in $\mathcal{S}^i$, the selected vertices correspond to the same edge of $G$, denote these edges by $e_1,e_2,e_3$, respectively. 
Furthermore, as in the proof of Claim~\ref{thm:ReductionK14-1}, we observe that the edges $e_1,e_2,e_3$ share a vertex $v \in V_i$. 
If $I_i$ is complete, we set $\phi(i) = v$. Otherwise, we set $\phi(i)$ arbitrarily.

Let $Q$ be the set of those $i$, for which $I_i$ is not complete.
We observe that each base clique $B$ is in at most three sets $\mathcal{S}^i$. Consider a base clique $B$.
If $B$ is a primary clique or a cycle clique, then it corresponds to some $E_{ij}$, and $B$ belongs $\mathcal{S}^i$ and $\mathcal{S}^j$. In the last case, if $B$ is an equality clique in the equality gadget joining edge-cycles corresponding to, say, $ij_1$ and $ij_2$, then $C$ belongs to $\mathcal{S}^i,\mathcal{S}^{j_1},\mathcal{S}^{j_2}$.
Summing up, each base clique belongs to at most three sets $\mathcal{S}^i$.
Since there are fewer than $\gamma' \cdot k$ base cliques $B$, such that $B \cap I = \emptyset$, we observe that the size of $Q$ is at most $3\gamma' \cdot k$. The vertices in $Q$ are incident to at most $9\gamma' \cdot k$ edges in $J$, and all remaining edges are satisfied by $\phi$. So, using \eqref{eq:k-upperbound}, we obtain
\[
\val(\Gamma) \geq \frac{|E(J)|-9\gamma' \cdot k}{|E(J)|} \geq 1-45s \cdot \gamma'=1-\gamma.
\qedhere
\]
\end{proof}

\subsection{Refuting Constant-Factor FPT Approximation} \label{sec:noconstant}
In this section we prove \cref{thm:NoConstantFptApprox}.
However, as mentioned in \cref{sec:Intro}, we need to consider a larger class than $\cC^*([4,z],5)$ to obtain the lower bound.
Let ${\cal P}(a,b)$ be a graph family consisting of cycles $C_p$ for all $p \in [a,b]$ and all trees without twins in ${\cal T}(\lceil\frac{b - 1}{2}\rceil)$ and let $\mathcal{D}(a,b)$ be the class of ${\cal P}(a,b)$-free graphs.
Note that $\cC^*([a,b], c) \subseteq {\cal D}(a,b)$ as ${\cal P}(a,b) \subseteq \cH([a,b],c) \cup {\cal T}(\lceil\frac{b - 1}{2}\rceil)$.
We will prove the following theorem that implies \cref{thm:NoConstantFptApprox}.

\begin{theorem}
\label{thm:NoConstantFptApproxFormal}
Let $z \geq 5$ be a constant. Let $\gamma > 0$ be a constant and let $f: \mathbb{N} \to \mathbb{N}$ be a computable function.
Unless the deterministic Gap-ETH fails, there is no algorithm, given an $n$-vertex instance $G \in \mathcal{D}(5,z)$  and an integer $k$, runs in time $f(k) \cdot n^{\Oh(1)}$ and can distinguish between the two cases:
$\alpha(G) \geq k$, and $\alpha(G) < (1 - \gamma) \cdot k$.
\end{theorem}

The idea of the proof is to use the \emph{lexicographic product} to amplify the approximation factor given by statement \eqref{it:NoFptApprox-classes} of \cref{thm:noFPAS-classes}.
 Let $G_1 = (V_1, E_1)$ and $G_2 = (V_2, E_2)$ be graphs.
 The \emph{lexicographic product} $G_1 \times_\ell G_2$ is the graph $G = (V,E)$ such that $V = V_1 \times V_2$ and $(u_1,v_1)(u_2,v_2) \in E$ if $u_1u_2 \in E_1$ or $u_1 = u_2$ and $v_1v_2 \in E_2$.
 In other words, the graph $G$ consist of copies $G^u_2$ of $G_2$, one for each $u \in V_1$, and a vertex $v_1$ from $G^{u_1}_2$ and a vertex $v_2$ from $G^{u_2}_2$ (for $u_1 \neq u_2$) are adjacent if and only if $u_1u_2 \in E_1$.
 We use the following two properties of the lexicographic product to obtain our result.
 
 \begin{proposition}[Geller and Stahl~\cite{Geller75}]
  \label{prp:ProductIS}
For graphs $G_1, G_2$, it holds that $\alpha(G_1 \times_\ell G_2) = \alpha(G_1)\cdot \alpha(G_2)$.
 \end{proposition}
 
 Unfortunately, the lexicographic product does not preserve ``$H$-freeness'' for all graphs $H \in \cH([a,b],c) \cup {\cal T}(b')$. Indeed, it might contain a copy of $H$ even if the original graphs were $H$-free.
 However, this might happen only if $H$ has some specific structure, as shown in the next proposition.
 Note that no graph in $\cH([a,b],c)$ and ${\cal T}(b')$ contains a triangle and they are all connected.
 
 \begin{proposition}
  \label{prp:ProductHFree}
  Let $H$ be connected, triangle-free graph without twins.
  Let $G_1$ and $G_2$ be $H$-free graphs.
  Then, $G_1 \times_\ell G_2$ is also $H$-free.
 \end{proposition}
 \begin{proof}
  Suppose for a contradiction that $G = G_1 \times_\ell G_2$ contains $H$ as an induced subgraph.
  As we mentioned above, $G$ consists of copies $G^u_2$ of $G_2$ for each $u \in V_1$.
  First, the copy of $H$ cannot be completely contained in one copy $G^u_2$ as $G_2$ is $H$-free.
  Supose that each copy $G^u_2$ contains at most one vertex of $H$.
  Then, the graph $G_1$ would contain $H$ as an induced subgraph.
  Thus, there is a copy $G^{u_1}_2$ that contains at least two vertices of $H$, say $w_1 = (u_1,v_2)$ and $w_2 = (u_1,v_2)$.
  
  The graph $H$ has no twins and the neighbors of $w_1$ and $w_2$ outside of $G^{u_1}_2$ are the same.
  Thus, there is another vertex $w_3 = (u_1,v_3)$ of $H$ in $G^{u_1}_2$ such that $w_3$ is adjacent to one of the vertices $w_1$ and $w_2$, without loss of generality say $w_1$.
  Since the graph $H$ is connected and is not entirely contained in $G^{u_1}_2$, there is a vertex $w_4 = (u_2,v')$ of $H$ in $G^{u_2}_2$ such that $w_4$ is adjacent to at least one vertex of $w_1,w_2,w_3$.
  However, since at least one edge is present between $G^{u_1}_2$ and $G^{u_2}_2$, there is an edge $u_1u_2 \in E(G_2)$ and therefore, there is a complete bipartite graph between $G^{u_1}_2$ and $G^{u_2}_2$.
  Thus, $w_4$ is connected to all $w_1, w_2$, and $w_3$.
  Since $H$ is an induced subgraph of $G$, the graph $H$ would contain a triangle $w_1,w_3,w_4$, which is a contradiction.
 \end{proof}

 When we restrict the family $\cC([4,b],c)$ to the graphs without twins we get exactly a family consisting of cycles of length at least 5 and at most $b$, as cycles of length at least 5 do not contain twins and on the other hand the stars and $C_4$ contain twins.
 Hence, by restricting the family $\cC^*([4,z],5)$ (as used in \cref{thm:noFPAS-classes}) to the graphs without twins we obtain exactly the family ${\cal P}(5,z)$.
 Note that graphs in $\cC([5,z],4)$ without twins are in ${\cal P}(5,z)$ as well.
 
 \begin{proof}[Proof of \cref{thm:NoConstantFptApproxFormal}]
  Suppose for a contradiction there is a constant $\gamma_0 > 0$ and an algorithm $\mathcal{A}$ with runtime $f(k)\cdot n^c$ for a computable function $f$ and a constant $c$ that for an input graph $G \in \mathcal{D}(5,z)$ can distinguish between two cases whether $\alpha(G) \geq k$ or $\alpha(G) < (1 - \gamma_0)\cdot k$.
  Let $\gamma$ be a constant given by statement \eqref{it:NoFptApprox-classes} of \cref{thm:noFPAS-classes}.
  Recall that $\cC^*([4,z], 5) \subseteq \mathcal{D}(5,z)$.
  Thus in particular, there is no algorithm with runtime $f(k)\cdot n^{\Oh(1)}$ that can distinguish between the cases whether $\alpha(G) \geq k$ or $\alpha(G) < (1 - \gamma)\cdot k$ (under the deterministic Gap-ETH).

  Let $d$ be the smallest integer such that $(1 - \gamma)^d \leq (1 - \gamma_0)$.
  Now, let $G^{d}$ be a $d$-fold lexicographic product of $G$ with itself, i.e.,
  \[
   G^d = \underbrace{G \times_\ell \dots \times_\ell G}_d.
  \]
  Recall that each graph in ${\cal P}(5,z)$ is connected, triangle-free and without twins, thus by \cref{prp:ProductHFree}, $G^d \in \mathcal{D}(5,z)$ as well.
  Further by \cref{prp:ProductIS}, $\alpha(G^d) = \alpha(G)^d$.
  Now consider the two cases listed in the statement.
  If $\alpha(G) \geq k$, then $\alpha(G^d) \geq k^d$.
  On the other hand, if $\alpha(G) < (1 - \gamma)\cdot k$, then $\alpha(G^d) < (1 - \gamma)^d\cdot k^d \leq (1-\gamma_0)\cdot k^d$ by the definition of $d$.
  Thus, the algorithm $\mathcal{A}$ would distinguish the cases whether $\alpha(G^d) \geq k^d$ or $\alpha(G^d) < (1 - \gamma_0)\cdot k^d$ in time $f(k^d)\cdot n^c$.
  Subsequently, we can distinguish between the cases whether $\alpha(G) \geq k$ or $\alpha(G) \leq (1 - \gamma)\cdot k$ in time $f(k^d) \cdot (n^d)^c = f'(k) \cdot n^{\Oh(1)}$ for a computable function $f'$, which is a contradiction with statement \eqref{it:NoFptApprox-classes} of \cref{thm:noFPAS-classes}.
 \end{proof}

%% file: src/parameter-H.tex
In this section we still consider the \IS problem in $H$-free graphs, but now 
our parameter is related to the graph $H$. 
First, we show \cref{thm:parameterH}. We point out that a similar argument was also observed by Bonnet~\cite{Bonnet-private}.

\parameterH*

\begin{proof}
We will reduce from \MIS, for which the vertices of the input graph are 
partitioned into $k$ disjoint sets $V_1,V_2\ldots,V_k$, each of which forms a 
clique. Note that any independent set can contain at most one vertex from each 
set $V_i$ where $i\in\{1,\ldots, k\}$. Let $\mathcal{H}_d$ be a class of graphs 
as in the statement. Set $k=d$ and let $G$ be an instance of \MIS. Let us 
observe that the vertex set of $G$ is partitioned into $k=d$ cliques, so~$G$ is 
clearly $H$-free for every~$H \in \mathcal{H}_d$.

By simply taking the complement of the input graph, we can easily establish 
that \MIS is as hard as \MCSI{} where $J$ is a clique, i.e., the \MC problem. 
Thus \MIS is \Wone-hard and has no $n^{o(k)}$ algorithm, unless the ETH 
fails~\cite[Theorem 13.25 and Corollary 14.23]{DBLP:books/sp/CyganFKLMPPS15}. 
Furthermore, by a result of Lin et al.~\cite{lin_et_al:LIPIcs.ICALP.2022.90} the 
\MC problem has no $k^{o(1)}$-approximation in $f(k)n^{O(1)}$ time under ETH, 
and by Chalermsook et al.~\cite{DBLP:conf/focs/ChalermsookCKLM17} no clique of 
size $\zeta(k)$ can be computed in $f(k)n^{\zeta(k)}$ time under the 
deterministic Gap-ETH. From these results the statement follows.
\end{proof}

%% file: src/sparsification.tex
Now let us consider the \IS problem in $K_{1,d}$-free graphs, parameterized by \emph{both} $k$ and $d$.
In this 
case we are able to give parameterized approximation lower bounds based on the 
following sparsification of \MCSI{}. 

\begin{theorem}\label{thm:sparse-MCSI}
Consider an instance $\Gamma = (G,V_1,\dots,V_\ell,J)$ of \MCSI{$t$} with 
size~$n$. Let $\xi(\ell)=2^{(\log\ell)^{1/2+\eps}}$ for any constant 
$0<\eps<1/2$, and let $\zeta$ be any function in $\omega(1)$.
Given that $t>\xi(\ell)$ or~$t>\zeta(\ell)$, respectively, for any computable 
function $f$, there is no $f(\ell)\cdot n^{\Oh(1)}$ time algorithm that can 
distinguish between the two cases:
\begin{compactenum}
 \item (YES-case) $\val(\Gamma) = 1$, and 
 \item (NO-case) 
 \begin{compactitem}
 \item $\val(\Gamma) \leq \xi(\ell)/t$ assuming the deterministic Gap-ETH, and
 \item $\val(\Gamma) \leq \zeta(\ell)/t$ assuming the Strongish Planted Clique 
Hypothesis.
 \end{compactitem}
\end{compactenum}

\end{theorem}

To prove \cref{thm:sparse-MCSI} we need two facts. The first is the Erd{\H 
o}s-Gallai theorem on \emph{degree sequences}, which are sequences of 
non-negative integers $d_1,\ldots,d_n$, for each of which there exists a simple 
graph on $n$ vertices such that vertex $i\in[n]$ has degree $d_i$. We use the 
following constructive formulation due to Choudum~\cite{choudum86}.

\begin{theorem}[Erd\H{o}s-Gallai theorem~\cite{choudum86}]
\label{thm:erdosgallai}
A sequence of non-negative integers $d_1 \geq \dots \geq d_n$ is a degree 
sequence of a simple graph on $n$ vertices if $d_1 + \dots + d_n$ is even and 
for every $1 \leq k \leq n$ the following inequality holds:
\newline
$  \sum_{i = 1}^k d_i \leq k(k-1) + \sum_{i = k+1}^n \min(d_i, k).$
Moreover, given such a degree sequence, a corresponding graph can be 
constructed in polynomial time.
\end{theorem}

We also need parameterized approximation lower bounds for \MCSI{}, as 
given by Dinur and Manurangsi~\cite{DBLP:conf/innovations/DinurM18} and 
Manurangsi et al.~\cite{manurangsi_et_al:LIPIcs.ITCS.2021.10}.

\begin{theorem}[Dinur and Manurangsi~\cite{DBLP:conf/innovations/DinurM18}, 
Manurangsi et al.~\cite{manurangsi_et_al:LIPIcs.ITCS.2021.10}]
\label{thm:hard-detgapETH}
Consider an instance $\Gamma = (G,V_1,\dots,V_\ell,J)$ of \MCSI{} with 
size~$n$ and $J$ a complete graph. Let $\xi(\ell)=2^{(\log\ell)^{1/2+\eps}}$ for 
any constant $0<\eps<1/2$, and let $\zeta$ be any function in $\omega(1)$. There 
is no $f(\ell)\cdot n^{\Oh(1)}$ time algorithm for any computable function~$f$ 
that can distinguish between the following two cases: 
\begin{compactenum}
 \item (YES-case) $\val(\Gamma) = 1$, and
 \item (NO-case) 
 \begin{compactitem}
 \item $\val(\Gamma) \leq \xi(\ell)/{\ell}$ under the deterministic Gap-ETH, and
 \item $\val(\Gamma) \leq \zeta(\ell)/{\ell}$ under the Strongish Planted 
Clique Hypothesis.
 \end{compactitem}
\end{compactenum}
\end{theorem}

\begin{proof}[Proof of \cref{thm:sparse-MCSI}]
Let $\Gamma = (G,V_1,\dots,V_\ell,J)$ be an instance of \MCSI{} where~$J$ is a 
complete graph. To find an instance of \MCSI{$t$} given $\Gamma$, we first need 
to construct a graph $J'$ with maximum degree $t$, for which we use the Erd{\H 
o}s-Gallai theorem. For this, let $\ell'=\ell$ if $\ell$ is even and 
$\ell'=\ell-1$ if $\ell$ is odd. Now, by \cref{thm:erdosgallai} it is easy to 
verify that a $t$-regular graph on $\ell'$ vertices exists as $t\ell'$ is even. 
Moreover, the proof of \cref{thm:erdosgallai} by Choudum~\cite{choudum86} is 
constructive, so that we can compute~$J'$ on $\ell$ vertices in polynomial time 
by setting it to the constructed $t$-regular graph if $\ell'=\ell$, or by adding 
one more isolated vertex if~$\ell'=\ell-1$. Note that 
$V(J')=V(J)=\{1,\ldots,\ell\}$, $E(J')\subseteq E(J)$ as $J$ is a complete 
graph, and $|E(J')|=t\ell'/2$.

We create a graph $G'$ by removing edges from $G$ according to~$J'$. That is, 
for any $1\leq i,j\leq\ell$, if $ij \not \in E(J')$ then we remove all edges 
between sets $V_i$ and~$V_j$. The resulting subgraph of $G$ is called $G'$, and 
we get an instance $\Gamma' = (G',V_1,\dots,V_\ell,J')$ of~\MCSI{$t$}.

It is easy to see that if $\text{val}(\Gamma) = 1$, then $\text{val}(\Gamma') = 
1$ as well: we just use the optimal solution for $\Gamma$ and remove any edges 
non-existent in $G'$. Now suppose that $\text{val}(\Gamma) \leq \nu$, which 
means that every solution $\phi$ satisfies at most a $\nu$-fraction of edges of 
$J$. Let $\phi$ be an arbitrary solution of $\Gamma'$, which is also a solution 
for~$\Gamma$ as $G'\subseteq G$ and $J'\subseteq J$. By our assumption we know 
that it satisfies at most $\nu\cdot |E(J)|$ edges of~$J$. Thus, the solution 
$\phi$ satisfies at most $\nu\cdot |E(J)|$ edges of $J'$ as well, and we 
obtain
\[
\val(\Gamma') \leq \frac{\nu\cdot |E(J)|}{|E(J')|} = 
\nu\cdot\frac{\ell(\ell - 1)/2}{t\ell'/2} \leq 
\nu\cdot \frac{\ell(\ell - 1)}{t(\ell-1)}=
\nu\cdot\frac{\ell}{t}.
\]

By the first part of \cref{thm:hard-detgapETH}, no $f(\ell) \cdot 
n^{\Oh(1)}$ time algorithm can distinguish between $\val(\Gamma)=1$ and 
$\val(\Gamma)\leq {\xi(\ell)}/{\ell}$ given $\Gamma$, where $\xi(\ell)=2^{(\log 
k)^{1/2+\eps}}$ for any constant $0<\eps<1/2$, under the deterministic 
Gap-ETH. By the above calculations, for $\Gamma'$ we obtain that no such 
algorithm can distinguish between $\val(\Gamma')=1$ and $\val(\Gamma')\leq 
\xi(\ell)/t$ by setting $\nu={\xi(\ell)}/{\ell}$, and so we obtain the first 
part of \cref{thm:sparse-MCSI}.

When using the second part of \cref{thm:hard-detgapETH} instead, under the 
Strongish Planted Clique Hypothesis, given~$\Gamma$ and any function 
$\zeta\in\omega(1)$, no $f(\ell) \cdot n^{\Oh(1)}$ time algorithm can distinguish 
between $\val(\Gamma)=1$ and $\val(\Gamma)\leq {\zeta(\ell)}/{\ell}$. Analogous 
to before, we obtain the second part of \cref{thm:sparse-MCSI} by setting 
$\nu=\zeta(\ell)/\ell$.
\end{proof}

Based on \cref{thm:sparse-MCSI} we can prove \cref{thm:starsparsified} using 
the reduction of \cref{lemma:K1dfreereduction}. 

\starsparsified*

\begin{proof}
We reduce via \cref{lemma:K1dfreereduction} from \MCSI{$t$} to \IS, which given 
an instance $\Gamma$ of \MCSI{$t$} results in a $K_{1,2t+2}$-free graph $G$ for 
\IS. We thus set $d=2t+2$. If $\val(\Gamma)=1$, then $G$ has an 
independent set of size $k={\ell\choose 2}$. If $\val(\Gamma)\leq\xi(\ell)/t$ 
or $\val(\Gamma)\leq\zeta(\ell)/t$, then every independent set of $G$ has size 
at most $\xi(\ell){\ell\choose 2}/t\leq \frac{\xi(k)k}{d/2-1}$ or 
$\zeta(\ell){\ell\choose 2}/t\leq \frac{\zeta(k)k}{d/2-1}$, respectively, 
assuming w.l.o.g.\ that $k\geq 4$ so that $\ell\leq 2\sqrt{k}\leq k$. In the 
first case, given a constant $\eps'>0$ we may choose $\eps$ small enough in 
\cref{thm:sparse-MCSI} so that $\frac{\xi(k)k}{d/2-1}\leq 2^{(\log 
k)^{1/2+\eps'}} k/d$. Thus, for $\xi'(k)=2^{(\log k)^{1/2+\eps'}}$, a 
$d/\xi'(k)$-approximation algorithm for \IS would be able to distinguish between 
the YES- and NO-case of~$\Gamma$. In the second case, given any function 
$\zeta'\in\omega(1)$, we may choose an appropriate function $\zeta\in\omega(1)$ 
in \cref{thm:sparse-MCSI} for which $\frac{\zeta(k)k}{d/2-1}\leq \zeta'(k)k/d$. 
Thus a $d/\zeta'(k)$-approximation algorithm for \IS would be able to 
distinguish between the YES- and NO-case of~$\Gamma$.

Note that $d=2t+2\leq 2\ell$ as the maximum degree of the graph $J$ is $\ell-1$. 
Thus if the runtime of this algorithm is $f(d,k) \cdot n^{\Oh(1)}$, then for 
some function~$f'$ this would be a $f'(\ell) \cdot n^{\Oh(1)}$ time algorithm 
for \MCSI{$t$}. However, according to \cref{thm:sparse-MCSI} this would be a 
contradiction, unless the deterministic Gap-ETH or the Strongish Planted Clique 
Hypothesis fails, respectively. We may rename $\xi'(k)$ to $\xi(k)$ or 
$\zeta'(k)$ to $\zeta(k)$ to obtain \cref{thm:starsparsified}.
\end{proof}

%% file: src/conclusions.tex
Our parameterized inapproximability results of \cref{thm:noFPAS} suggest that 
the \IS problem is hard to approximate to within some constant, whenever it is 
\Wone-hard to solve on $H$-free graphs, according to
\cref{thm:bonnet-w1hard}. In most cases it is unclear though whether any 
approximation can be computed (either in polynomial time or by exploiting the 
parameter $k$), which beats the strong lower bounds for polynomial-time 
algorithms for general graphs. The only known exceptions to this are the 
$K_{1,d}$-free case, where a polynomial-time 
$(\frac{d-1}{2}+\delta)$-approximation algorithm was shown by 
Halld\'orsson~\cite{DBLP:conf/soda/Halldorsson95}, and the $K_{a,b}$-free case, 
for which we showed a polynomial-time 
$\Oh\bigl((a+b)^{1/a} \cdot \alpha(G)^{1-1/a}\bigr)$-approximation algorithm in \cref{thm:Kab}. 
For $K_{1,d}$-free graphs, we were also able to show an almost asymptotically 
tight lower bound for polynomial-time algorithms in \cref{thm:k1dfree}. For 
parameterized algorithms, our lower bound of \cref{thm:starsparsified} for 
$K_{1,d}$-free graphs does not give a tight bound, but seems to suggest that 
parameterizing by $k$ does not help to obtain an improvement.

Settling the question whether $H$-free graphs admit better approximations to \IS 
than general graphs, remains a challenging open problem, both for 
polynomial-time algorithms and algorithms exploiting the parameter $k$.

Let us point out one more, concrete open question. Recall from  \cref{thm:bonnet-w1hard}
Bonnet {\em et al.}~\cite{DBLP:conf/iwpec/BonnetBCTW18} were able to show \Wone-hardness
for graphs which \emph{simultanously} exclude $K_{1,4}$ and all induced cycles of length in $[4,z]$, for any constant $z \geq 5$.
On the other hand, we presented two separate reductions, one for $(K_{1,5},C_4,\ldots,C_z)$-free graphs,
and another one for $(K_{1,4},C_5,\ldots,C_z)$-free graphs.
It would be nice to provide a uniform reduction, i.e., prove hardness for 
parameterized approximation in  $(K_{1,4},C_4,\ldots,C_z)$-free graphs.

Finally, note that the statement \eqref{it:xp-gap} of \cref{thm:MCSIHardness} only excludes algorithms with running time $f(k) \cdot n^{o(\sqrt{k})}$. However, a straightforward algorithm has running time $f(k) \cdot n^{\Oh(k)}$. Is is possible to obtain a matching lower bound (at least up to polylogarithmic factors in the exponent)? 